\documentclass[journal]{IEEEtran}

\usepackage{fancyhdr}
\usepackage{lipsum}
\usepackage[table]{xcolor}
\usepackage{flushend}
\usepackage{amsmath, amssymb}
\usepackage[noend]{algpseudocode}
\usepackage{algorithm}
\usepackage{setspace} 
\usepackage{graphicx}   
\usepackage{varwidth}
\usepackage{enumitem}
\usepackage{xspace}
\usepackage{multirow}
\usepackage{enumitem}
\usepackage{caption}
\usepackage{subcaption}
\usepackage{booktabs}
\usepackage{colortbl}
\usepackage{flushend}
\usepackage{balance}

\usepackage{caption}
\usepackage{amsthm}
\usepackage{listings}
\newtheorem{theorem}{Theorem}
\newtheorem{lem}{Lemma}

\newcommand{\Ab}{{\bf A}\xspace}
\newcommand{\xb}{{\bf x}\xspace}
\newcommand{\algrule}[1][.2pt]{\par\vskip.5\baselineskip\hrule height #1\par\vskip.5\baselineskip}

\begin{document}

\title{Fully-Automated Code Generation for Efficient Computation of Sparse Matrix Permanents on GPUs}

\author{Deniz~Elbek
        and~Kamer~Kaya
\thanks{Elbek (deniz.elbek@sabanciuniv.edu) and Kaya (kaya@sabanciuniv.edu) are with the Department
of Computer Science and Engineering at the Faculty of Engineering and Natural Sciences, Sabancı University, Istanbul, Turkey.}
}


\maketitle

\begin{abstract} Registers are the fastest memory components within the GPU's complex memory hierarchy, accessed by names rather than addresses. They are managed entirely by the compiler through a process called \textit{register allocation}, during which the compiler attempts to cache predictable data from thread-local memory into thread-private registers. Computing the permanent of a sparse matrix poses a challenge for compilers, as optimizing this process is hindered by the unpredictable distribution of nonzero elements, which only become known at runtime. In this work, we employ \textit{fully-automated code generation} to address this, producing highly optimized kernels tailored to the matrix's sparsity pattern. State-of-the-art permanent computation algorithms require each thread to store a private array, denoted $\xb$, of size $n$. We first propose a technique that fully stores these arrays in registers, with \textit{inclusion} and \textit{exclusion} kernels generated for each column. To \textit{minimize control divergence} and reduce the number of unique kernels within a warp, we exploit the internal structure of Gray codes, which are also used in the state-of-the-art algorithm. Our second technique reduces register pressure by utilizing both registers and global memory and introduces a matrix ordering and partitioning strategy for greater efficiency. On synthetic matrices, this approach achieves a $31\times$ speedup over state-of-the-art CPU implementations on 112 cores, and an $8\times$ speedup compared to our traditional GPU implementation. For real-world matrices, these speedups are $24.9\times$ and $4.9\times$. \end{abstract}


\IEEEpeerreviewmaketitle

\section{Introduction}
\IEEEPARstart{T}{he} \textit{permanent} of a matrix is used to capture its characteristic properties and is applied in various domains including quantum computing~\cite{aaronson11, LUNDOW2022110990,10.1093/nsr/nwy079,Brod15}, complexity theory~\cite{Narahara2013ApplicationOP, PhysRevE.77.016706}, and probability and statistics~\cite{kilic2007, balakrishnan2007, Merschen11}. It has also been employed in graph theory~\cite{lovász2009matching, 253457, DUFOSSE2022130}, in which the data structures at hand are often sparse. The permanent belongs to the \textit{immanant} family~\cite{LittlewoodRichardson1934, Littlewood1950}, which also includes the determinant. While the permanent is similar to the determinant, its formula lacks the alternating sign in the summation, making its computation immensely challenging. The permanent of an \(n \times n\) matrix \(\mathbf{A} = [a_{ij}]\) is:
\begin{equation}
\text{perm}(\mathbf{A}) = \sum_{\sigma \in S_n} \prod_{i=1}^{n} a_{i,\sigma(i)}
\label{eq:naive}
\end{equation}
where the outer sum \(\sum_{\sigma \in S_n}\) runs over all permutations \(\sigma\) of the set \(S_n  = \{1, 2, \dots, n\}\), and the inner product \(\prod_{i=1}^{n} a_{i,\sigma(i)}\) multiplies the matrix entries \(a_{i,\sigma(i)}\). This formula has a time complexity of \(\Theta(n \times n!)\), making it infeasible to use even for a small, e.g., \(20 \times 20\), matrix. 

The computational complexity of the permanent is shown to be \#P-complete~\cite{valiant79a}, with the most efficient known algorithm having a time complexity of $O(2^{n-1}n)$ for an $n \times n$ matrix~\cite{ryser63,nijenhuis78}. Recently, HPC clusters and supercomputers have been used in practice to compute permanents for boson sampling; Wu~et~al. reported that the permanent of a $48 \times 48$ dense matrix is computed on $8192$ nodes of the Tianhe-2 supercomputer ($196,608$ cores) in $4500$ seconds. They also tested a hybrid approach using both CPU cores and two MIC co-processors/node with 256 nodes to compute the permanent of a $40 \times 40$ matrix in 132 seconds~\cite{10.1093/nsr/nwy079}. In another study, Lundow and Markstr\"{o}m used a cluster to compute the permanent of a record-breaking $54 \times 54$ matrix in 7103 core/days~\cite{LUNDOW2022110990}. As discussed in this study, and as well as in~\cite{Brod15} by Brod, low-depth Boson sampling set-ups lead to sparse matrices.\looseness=-1

In this work, we address the problem of computing permanents for sparse matrices using GPUs. Sparse matrix algorithms are challenging in efficiently utilizing GPU memory hierarchies, especially registers, due to the unpredictable distribution of nonzeros. Instead, our approach leverages the nonzero distribution to generate matrix-specific codes. The contributions of this work are summarized below:\looseness=-1
\begin{itemize}[leftmargin=*]
\item By leveraging matrix-specific, fully-automated code generation, we propose optimized GPU kernels that make better use of memory resources while minimizing control divergence through the properties of Gray codes. 
\item Additionally, we introduce a hybrid-memory approach that alleviates the pressure on registers by selectively using global memory, which, in combination with matrix ordering and partitioning strategies, leads to significant performance improvements. 
\item Our techniques yield $31\times$ and $8\times$ speedup compared to state-of-the-art CPU/GPU implementations, respectively, on synthetic matrices. On real-life sparse matrices, these values are $24.9\times$ and $4.9\times$, respectively. 
\end{itemize}

The rest of this paper is organized as follows: In Section II, we review the necessary background and notation, followed by a detailed description of our code generation approach in Section III. Section IV introduces our strategy for minimizing control divergence, while Section V presents the hybrid-memory code generation technique. Experimental results demonstrating the performance benefits of our approach are presented in Section VI, and related work is discussed in Section VII. Finally, Section VIII concludes the paper and outlines future research directions.

\section{Background and Notation}
Ryser~\cite{ryser63} modified the permanent equation using the inclusion-exclusion principle, enabling~\eqref{eq:naive} to be rewritten as:
\begin{equation}
\text{perm}(\mathbf{A}) = (-1)^n \sum_{S \subseteq \{1, 2, \dots, n\}} (-1)^{|S|} \prod_{i=1}^{n} \sum_{j \in S} a_{ij}
\label{eq:ryser}
\end{equation}
where the outer sum runs over all subsets \(S \subseteq \{1, 2, \dots, n\}\), \(|S|\) is the cardinality of the subset \(S\), and the inner product multiplies the sums for all rows \(i\), i.e., \(\sum_{j \in S} a_{ij}\) over the columns \(j \in S\). This equation has a time complexity of \(\Theta(2^{n} \times n^2)\). Nijenhuis and Wilf~\cite{nijenhuis78} further reduced the complexity to \(\Theta(2^{n-1} \times n)\) by going only over the subsets \(S \subseteq \{1, 2, \dots, n-1\}\) and processing the columns in Gray code order.\looseness=-1 

Gray code is a sequence where each successive binary code differs from the previous one by exactly one bit. This allows the algorithm to process only a single column and update the sums. An array \(\mathbf{x}\) is used to store the sums where the \(i\)-th element is the sum of the elements in the \(i\)-th row that correspond to the columns currently included in the subset $S$. When a column is added to or removed from $S$~(as indicated by the changed bit in the Gray code), only the affected column's contribution to the relevant row sums needs to be updated in \(\mathbf{x}\). This efficient update process removes the need for recalculating all the summations \(\sum_{j \in S} a_{ij}\) in \eqref{eq:ryser}, although it requires storing an \(\mathbf{x}\) array of size \(n\). We refer the reader to~\cite{nijenhuis78} for further details on the Nijenhuis-Wilf variant. 

For sparse matrices, the data structures {\em Compressed Sparse Row}~(CSR) and {\em Compressed Sparse Column}~(CSC) have been used for efficient permanent computation~\cite{kaya19}. For an $m \times n$ matrix with $\nu$ nonzeros, the CSR can be defined as follows:
\begin{itemize}[leftmargin=*]
    \item {\tt rptrs[$\cdot]$} is an integer array of size $m + 1$. For $0 \leq i < m$, {\tt rptrs[i]} is the location of the first entry of $i$th row in {\tt cids}. The first element is {\tt rptrs[0]} $= 0$, and the last element is {\tt rptrs[m]} $ = \nu$. Hence, all the column indices of row $i$ are stored between {\tt cids[rptrs[i]]} and {\tt cids[rptrs[i + 1]] - 1}. 
    \item {\tt cids[$\cdot$]} is an integer array storing the column IDs for each nonzero in row-major ordering. 
    \item Each nonzero value $val$ is stored in an array named {\tt rvals} in the same order of nonzeros inside {\tt cids}.
\end{itemize}
The CSC is simply the CSR of the transposed matrix where the arrays are {\tt cptrs}, {\tt rids} and {\tt cvals}, respectively. Using these sparse formats, Algorithm~\ref{alg:sparyser} implements the Nijenhuis-Wilf variant of~\eqref{eq:ryser} via Gray codes. In the algorithm, {\sc{Gray}}$_{g}$ denotes the $g$th Gray code and hence, $\log_2$({\sc{Gray}}$_{g}$ $\oplus$ {\sc{Gray}}$_{g-1}$) extracts the index of the changed bit, i.e., the column ID to process, for each iteration $1 \leq g < 2^{n-1}$. The {\bf for} loop at line~\ref{ln:outer} corresponds to the outer sum in~\eqref{eq:ryser} where the subsets $S \subseteq \{1,2, \ldots, n-1\}$, which denote the column IDs used to compute the row sums stored in \xb, are processed in Gray code order following the optimizations from Nijenhuis-Wilf~\cite{nijenhuis78}.  The sequential implementation of this variant with CRS/CCS structures is given in Algorithm~\ref{alg:sparyser}.

\begin{algorithm}[htbp]
\setstretch{1.08}
\caption{: {\sc SparsePerman}} 
\label{alg:sparyser}
\begin{flushleft}
\small
\textbf{Input:} ($rptrs, cids, rvals$) \algorithmiccomment{\small CSR of $n \times n$ sparse matrix $\Ab$} \\ \hspace*{7.2ex}($cptrs, rids, cvals$) \algorithmiccomment{CSC of $n \times n$ sparse matrix $\Ab$}\\
\textbf{Output:} {\tt perm} ($\Ab$) \algorithmiccomment{\small Permanent of $\Ab$}
\end{flushleft}

\algrule
\begin{algorithmic}[1]

\For{$i = 1$ to $n$} \algorithmiccomment{\small Nijenhuis-Wilf $\xb$ init.~\cite{nijenhuis78}} \label{ln:xinitstart}	
	\State{$sum = 0$}  
	\For{$ptr = rptrs[i]$ to ($rptrs[i+1] - 1$)}  
		\State{$sum \leftarrow sum + rvals[ptr]$}
	\EndFor
	\State{$\xb[i] \leftarrow rvals[rptrs[i+1] - 1]- \frac{sum}{2}$}  \label{ln:xinitend}	
\EndFor

 \algrule
	\State{$p \leftarrow \prod_{i = 1}^n{\xb[i]}$} \algorithmiccomment{\small Nijenhuis-Wilf result init.~\cite{nijenhuis78}}	

 \algrule
 \For {$g = 1$ to $2^{n-1} - 1$} \label{ln:outer}
 
      \hrulefill
   \State{$\blacktriangleright$ \footnotesize Find the changing column and the direction of change (+ or -)}
	\State{$j \leftarrow \log_2$({\sc{Gray}}$_{g}$ $\oplus$ {\sc{Gray}}$_{g-1}$)}\label{ln2:j}
	\State{$s \leftarrow 2 \times ${\sc{Gray}}$_{g}$[$j$] $-1$}\label{ln:s}

\hrulefill
 	  \State{$\blacktriangleright$ \footnotesize Update row sums by adding/removing column $j$ to/from $S$}

	  \For {$ptr = cptrs[j]$ to ($cptrs[j + 1] - 1$)}\label{ln:sumfor}
	        \State{$row \leftarrow rids[ptr]$}
	        \State{$val \leftarrow cvals[ptr]$}

	        \State{$\xb[row] \leftarrow \xb[row] + (s \times val)$}\label{ln:xupdate}
	 \EndFor
  
  \hrulefill
 	  \State{$\blacktriangleright$ \footnotesize Compute the inner product (of row sums)}

	 	\State{$prod \leftarrow 1$}\label{ln:reduce1}  	 
	 	\For {$i =  1$ to $n$} 
	        		\State{$prod \leftarrow prod \times \xb[i]$}   \label{ln:reduce2}  	    
	 	\EndFor
   
     \hrulefill
    	  \State{$\blacktriangleright$ \footnotesize Perform an outer sum iteration and update the result variable}

		\State{$p \leftarrow p + \left((-1)^g \times prod\right)$}
\EndFor 

\algrule
\State{\Return{$p \times (4 \times (n \bmod 2) - 2)$}} 
 \end{algorithmic}
\end{algorithm}

Figure~\ref{fig:gray_code_change} shows a toy $6 \times 6$ matrix that will be used throughout the rest of the paper as a running example. Algorithm~\ref{alg:sparyser} first initializes the $\xb$ vector corresponding to the sums of the rows for the Nijenhuis-Wilf variant~(lines 1--5). The cost of this part is $\mathcal{O}$($n + \nu$). Then, the main, compute-heavy body of the algorithm~(lines 7--21) starts. On the right of Fig.~\ref{fig:gray_code_change}, the ($s$, $j$)  pair for the first 8 iterations are given along with one more for column $j = 4$. For instance, the Gray code changes from $000000$ to $000001$ in the first iteration. Hence, the first bit corresponding to the first column, i.e., $j = 0$, changes, and furthermore, it changes from $0 \rightarrow 1$ which implies column {\em inclusion}, i,e, $s = 1$. Hence, the operations $\xb[0] \mathrel{+}= 11.6$, $\xb[2] \mathrel{+}= 2.6$, $\xb[3] \mathrel{+}= 1.8$, and $\xb[5] \mathrel{+}= 9.9$ are performed, i.e., column 0 is included on top of $\xb$~(lines 12--15). Column 0 will again be used in the third iteration, but this time with $s = -1$, since the Gray code changes from $000011$ to $000010$ and the changed bit changes from $1 \rightarrow 0$. After $\xb$ is updated, its contribution is computed by multiplying its entries~(lines 17--19) and added on top of the result variable $p$~(line 21). This process is repeated for $2^{n-1} - 1$ iterations. 

\begin{figure}[htbp]
    \centering
    \includegraphics[width=0.42\textwidth]{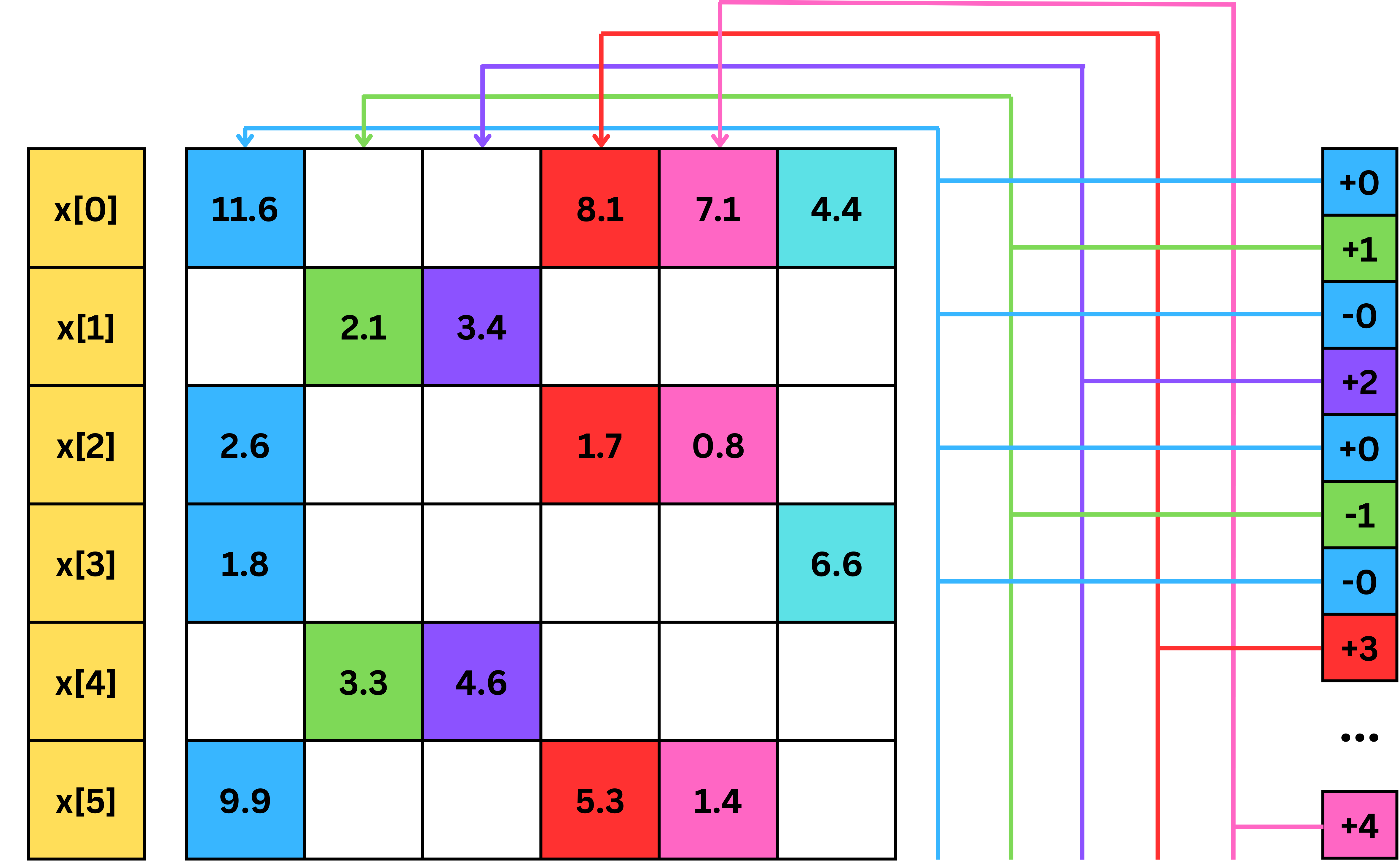}
    \caption{\small A 6x6 matrix (left), with a permanent value 54531.03, and the Gray-code change bit sequence (right) in the iterations of Algorithm~\ref{alg:sparyser} where each prefix `$+$`, indicating column inclusion, or `$-$`, indicating exclusion, correspond to the $s$ value, 1 or -1, at line 10 of Alg.~\ref{alg:sparyser}. Each number on the right is the changing bit index corresponding to the $j$ value at line 9 of Alg.~\ref{alg:sparyser}.}
    \label{fig:gray_code_change}
\end{figure}

\subsection{Parallel Permanent Computation}\label{subsec:parper}
When used for permanent computations, both~\eqref{eq:naive} and~\eqref{eq:ryser} are pleasingly parallelizable since different threads can independently execute the outer sum iterations. The Nijenhuis-Wilf variant, however, incurs a dependency between two adjacent iterations, as the \(\mathbf{x}\) array, which carries row sums, needs to be updated using the values in column $j$ where $j$ is the changed Gray code bit. A parallelization strategy with $\tau$ threads from~\cite{kaya19} assigns chunks of $\Delta = \lceil \frac{2^{n-1} - 1}{\tau} \rceil$ consecutive iterations to each thread~(may not be exact for the last thread). Hence, each thread $0 \leq t < \tau$ starts from the iteration $g^{start}_t = t \Delta + 1$ and stops at $\min(2^{n-1} - 1, g^{start}_t + \Delta - 1)$. With this approach, each thread needs to keep a {\em private} $\xb_t$ to accumulate the row sums where each $\xb_t$ is initialized from $\xb$ computed between lines~\ref{ln:xinitstart}--\ref{ln:xinitend} of Alg.~\ref{alg:sparyser} and is updated by adding the columns corresponding to the set bits in {\sc{Gray}}$_{g^{start}_t - 1}$. To reduce the synchronization overhead, each thread also keeps a partial permanent value which is added to a global variable after the iterations are completed.\looseness=-1  

\section{Computing Sparse Matrix Permanents on GPUs}\label{sec:gpuper}

While computing sparse matrix permanents is well-studied on CPUs, we are not aware of a study focusing on GPUs. Indeed, using the approach mentioned above, one can distribute the chunks to the GPU threads in a straightforward way. Considering the number of threads is much larger on GPUs~(compared to CPUs), the main overhead arises from the need to store a private \(\mathbf{x}\) array for each thread. 
As seen in Alg.~\ref{alg:sparyser}, $\Ab$~(in sparse format) and $\xb$ are the two main data structures accessed and/or modified during the execution. In this work, our focus will be data/memory-centric since the {\em arithmetic intensity}, i.e., the ratio of the number of FLOPs to the number of memory accesses to compute the permanent, is small which makes the kernel memory bound.

GPUs possess a diverse and hierarchical memory subsystem, where the choice of data storage significantly impacts the computation's execution time. Although both $\Ab$ and \(\mathbf{x}\) could be stored in {\em global memory}, this would incur a performance bottleneck. 
Instead, when {\em shared memory} is used, the performance is expected to increase by reducing data access times though at the cost of fewer threads being launched due to the existence of thread-private $\xb$ arrays stored on the limited shared memory resources. 
Table~\ref{tab:shared_vs_global} reports the execution time, the thread count recommended by the CUDA Occupancy API \cite{NVIDIA}, and the arithmetic intensity (w.r.t. to the number of bytes transferred from {\em global} memory of two different kernels which only differ in where they store \(\mathbf{x}\) on the GPU. As Table~\ref{tab:shared_vs_global} shows, the kernel \(\mathbf{x}_{shared}\) is $12.5\times$ faster than \(\mathbf{x}_{global}\), even the latter is implemented to perform coalesced global memory accesses and the former is launched with less number of, in fact $4\times$ less, threads. Nevertheless, being much faster, \(\mathbf{x}_{shared}\) is just an adaptation of the literature to GPU, and we consider it as one of our baselines.\looseness=-1 

\begin{table}[h!]
\setstretch{1.1}

\small
\centering
\begin{tabular}{|l|r|r|}
\hline
\textbf{Metric}            & $\xb_{shared}$ & $\xb_{global}$  \\ \hline
Execution time~(in sec)    & 1476.3           & 18449.4         \\ \hline
Thread count~(in the grid)               & 20480             & 122880           \\ \hline
Arithmetic intensity       & \multirow{2}{*}{$1.07 \times 10^8$} & \multirow{2}{*}{$1.25 \times 10^{-1}$} \\ 
(w.r.t. global memory)        & & \\ \hline

Speedup (w.r.t. $\xb_{global}$)                    & 12.5$\times$ & 1$\times$            \\ \hline
\end{tabular}
\caption{\small{Comparison of multiple metrics between two kernels, $\xb_{shared}$ and $\xb_{global}$, which store $\xb$ on shared and global memory, respectively, while computing the permanent of a random $45 \times 45$ sparse matrix on NVIDIA GV100.}}
\label{tab:shared_vs_global}
\end{table}

A powerful yet often underrated method for optimizing GPU performance is to fully exploit the registers, the fastest memory units available. A GPU has a huge register file that can provide almost immediate data access. For instance, an A100 GPU has 108 streaming multiprocessors~(SM), each accommodating 65,536 32-bit registers resulting in a capacity of around 28MB. Although registers are usually the designated temporary storage, in this work, we utilize them as permanent storage units which is especially challenging for sparse data. The data stored in a register cannot be accessed by other threads—except when the threads within a warp shuffle data. Consequently, any data stored in registers must be duplicated for every thread using them. For instance, using registers to store $\Ab$ results in duplications which will be extremely inefficient. However, for thread-private data such as $\xb_t$ for thread $t$, using registers can be promising.\looseness=-1 

\begin{figure*}[htbp]
     \centering
     \fbox{
     \begin{subfigure}[b]{0.31\textwidth}
         \centering
         \includegraphics[width=\textwidth]{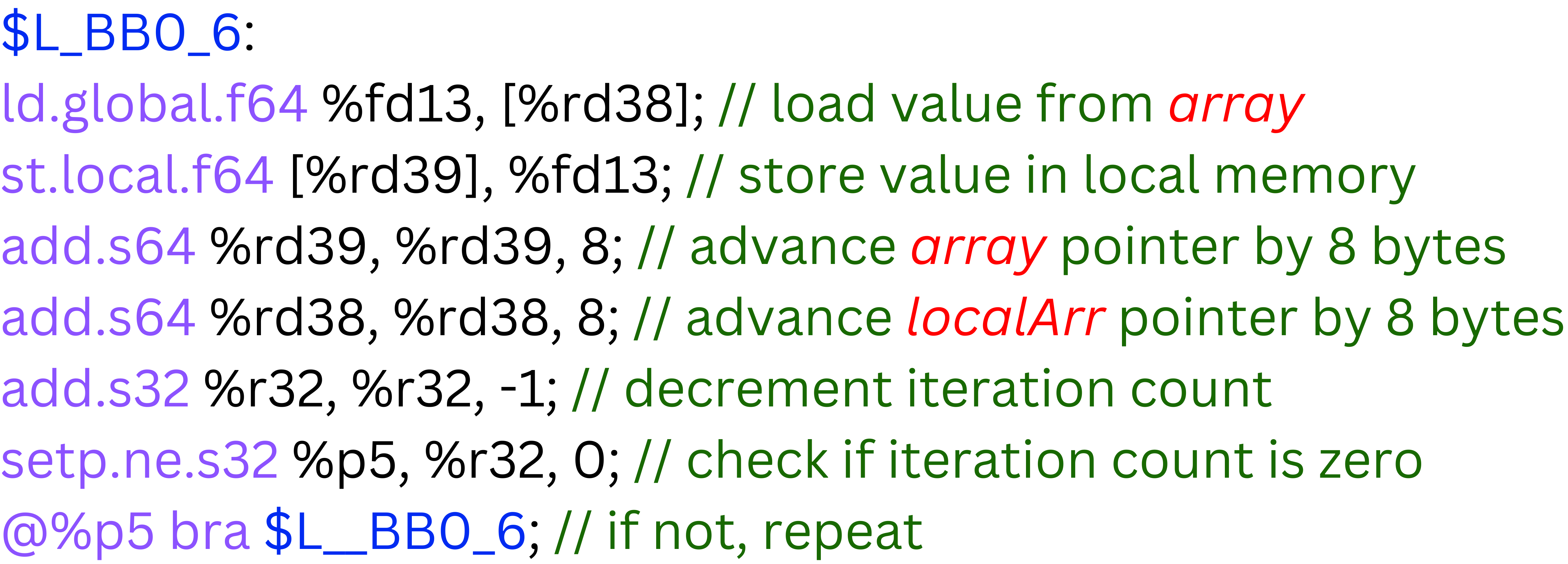}
         \caption{\small The copy loop}
         \label{fig:copyloop}
     \end{subfigure}
     \begin{subfigure}[b]{0.31\textwidth}
         \centering
         \includegraphics[width=\textwidth]{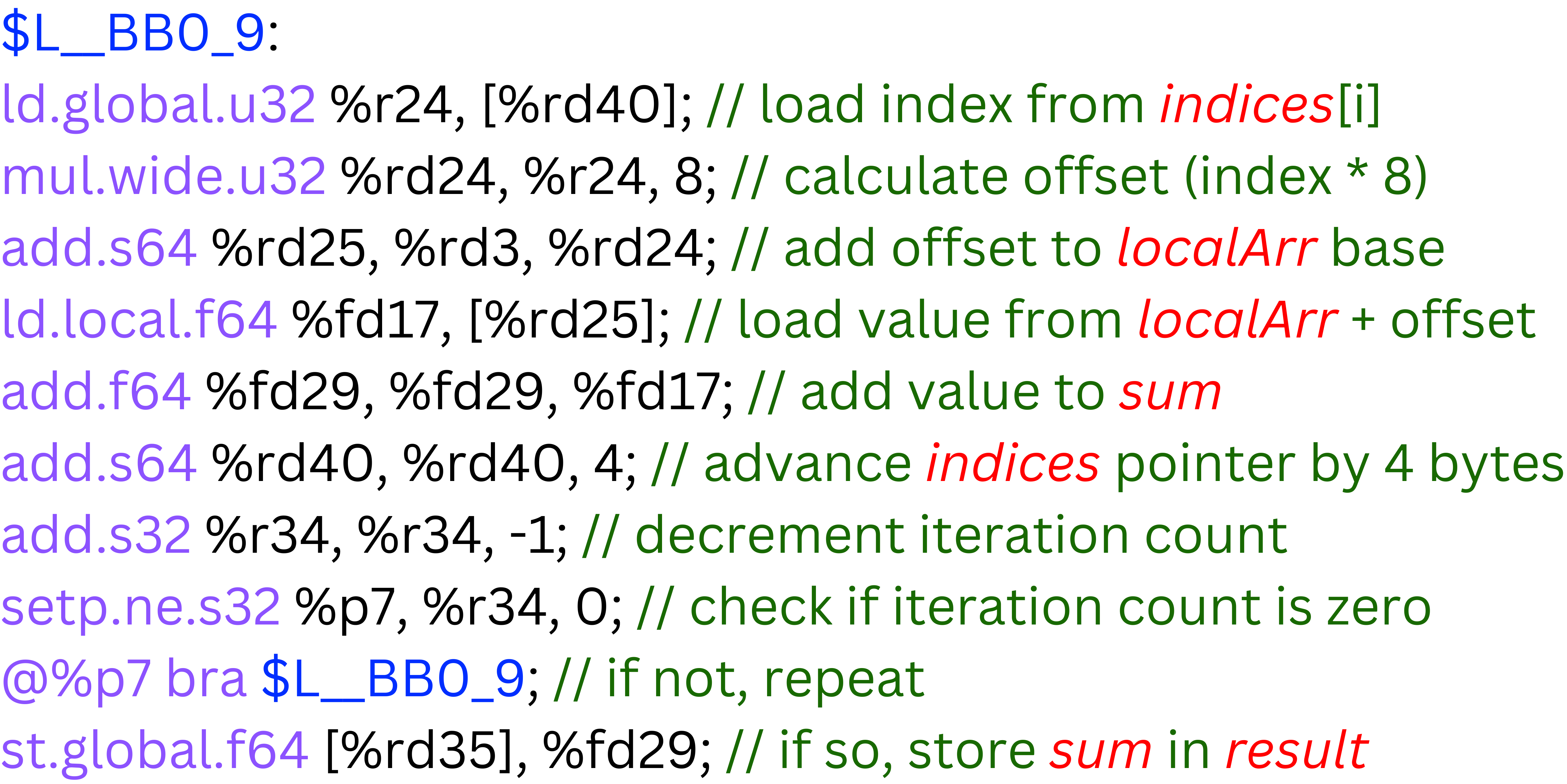}
         \caption{\small The reduce loop}
         \label{fig:reduceloop}
     \end{subfigure}
     }
     \fbox{
     \begin{subfigure}[b]{0.298\textwidth}
         \centering
         \includegraphics[width=\textwidth]{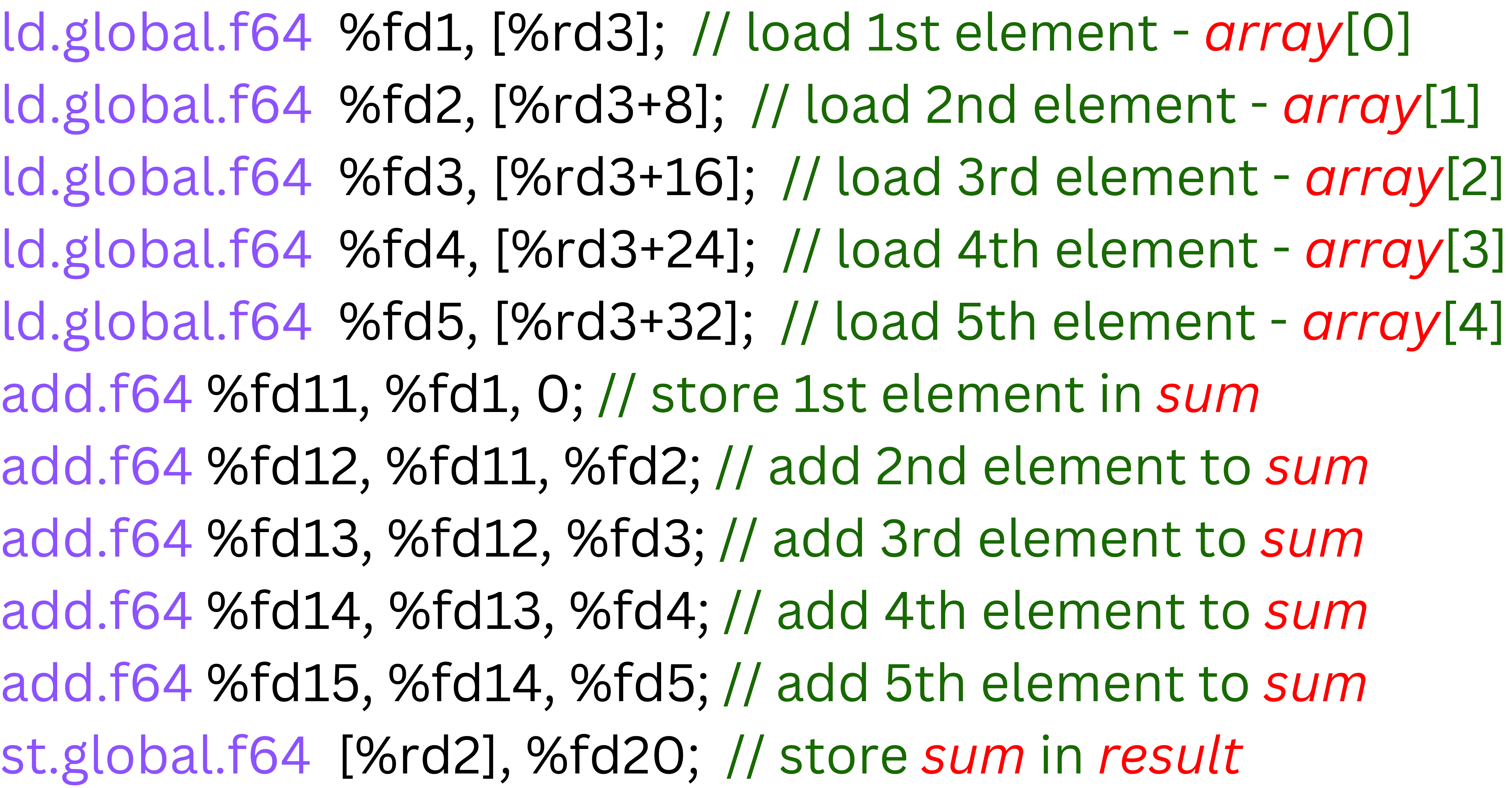}
         \caption{\small Copy/reduce loops with no indices}
         \label{fig:copyreduceloop}
     \end{subfigure}
     }
        \caption{\small Simplified PTX codes generated for Listing~\ref{lst:locReduce}. In the last version, the {\tt indices} parameter is assumed to be not presented and the reduction is performed sequentially over {\tt localArr}.  }
        \label{fig:allcodes}
\end{figure*}

\subsection{Register Allocation}
The CUDA runtime allocates local memory for each thread upon kernel launch. Access to this region is expensive, as it is logically located in the same region as the global memory. However, the data in this region is expected to be frequently referenced throughout the kernel execution. This contradiction leads the compiler to perform an important optimization called \textit{Register Allocation}, during which the data located in the thread-local memory is cached into thread-private registers. This allows threads to have fast access to the data they use most frequently. However, as dictated by the CUDA Programming Guide~\cite{NVIDIA}, two conditions must be met for this optimization: (1) data should fit into registers, and (2) data, e.g., the register to be accessed, should be discoverable at compile time.\looseness=-1 

For a single thread, the maximum number of registers is limited to 255. Therefore, the maximum data that can be cached into registers is around 1KB. The matrices the literature focuses on for Boson sampling have $n \leq 60$~\cite{LUNDOW2022110990}, and even in 64-bit precision, the maximum number of 32-bit registers per thread is not exceeded. Furthermore, when $n$ is larger, partial register-based storage of thread-private $\xb$ arrays is still possible. Hence, the first condition is met. Although the second condition is automatically met for dense matrices, for sparse matrices, the column-wise nonzero structure, {\tt \small cptrs} and {\tt \small rids}, is a part of the input and is not known at compile time. Hence, the indices of the $\xb$ entries accessed for the selected column~(at line 15 of Alg.~\ref{alg:sparyser}) are not known before the matrix is given.\looseness=-1

\begin{center}
\begin{tabular}{c}
\begin{lstlisting}[caption={\small A toy C++ code snippet for reduction.}, linewidth=8cm, numbers=left,
label={lst:locReduce}]
#define ARR_SIZE 5
__global__ void reduce( double* array, 
                        double* result, 
                        unsigned* indices) {
    double localArr[ARR_SIZE];
    for (int i = 0; i < size; ++i) {
        localArr[i] = array[i];
    }
    double sum = 0;
    for (int i = 0; i < ARR_SIZE; ++i) {
        sum += localArr[indices[i]];
    }
    *result = sum;
}
\end{lstlisting}
\end{tabular}
\end{center}

As an example of the problem mentioned above, consider the {\tt \small reduce} kernel provided in Listing~\ref{lst:locReduce} performing a reduction of 5 elements from {\tt \small localArr}. The elements are located by {\tt \small indices} provided as a parameter and {\em unknown} before the execution. Figs.~\ref{fig:allcodes}.(a) and~\ref{fig:allcodes}.(b) shows, respectively, the simplified PTX~(Parallel Thread Execution) codes for the copy~(lines 6--8) and the reduce~(lines 10--12) loops. Since the values in {\tt \small indices} are not known, the second condition mentioned above is not met and {\tt \small localArr} is not stored inside the registers. However, when line 11 of Listing~\ref{lst:locReduce} is modified as {\tt \small sum += localArr[i]}, i.e., when {\tt \small indices} are not used, the PTX code given by Fig.~\ref{fig:allcodes}.(c) is generated and {\tt \small localArr} is fully stored inside the registers, as the registers to be accessed at runtime are exactly known at compile-time.\looseness=-1

\subsection{Automated Code Generation for Matrix Permanents}
The Gray-code-order Algorithm~\ref{alg:sparyser} follows implies that for each iteration, the subset $S$ of columns differs from the previous set by the inclusion/exclusion of only a single column where $s$ at line~\ref{ln:s} of Alg.~\ref{alg:sparyser} being 1/-1 handles these two separate cases. When the matrix is known, the registers each column modifies, and the change direction for each iteration, i.e., addition or subtraction, become {\em known} at compile time which yields the possibility of storing $\mathbf{x}$ inside registers by meeting the second condition for Register Allocation.

\begin{center}
\begin{tabular}{c}
\begin{lstlisting}[caption={\small Inclusion kernel for column 0.}, numbers=left, linewidth=8.1cm, label={lst:colKernel}]
#define C double
__device__ __inline__ void c0_inc(C& product, 
    C& reg0,          const C& reg1, 
    C& reg2, C& reg3, const C& reg4, C& reg5) {
    reg0 += 11.600000;
    reg2 += 2.600000;
    reg3 += 1.800000;
    reg5 += 9.900000;

    prodReduce(product, reg0, reg1, 
                        reg2, reg3, reg4, reg5);
}
\end{lstlisting}
\end{tabular}
\end{center}

In this work, we first generate two {\em {\tt \small \_\_device\_\_} kernels} for each column, except the last one~(which is omitted due to Nijenhuis-Wilf optimization). These kernels are called {\em inclusion kernels}, handling $s = 1$, and {\em exclusion kernels}, handling $s = -1$, and are used to update the row sums in $\mathbf{x}$. As an example, the inclusion kernel for the column $j = 0$ of the toy matrix in Fig.~\ref{fig:gray_code_change} is given in Listing~\ref{lst:colKernel}. Each generated kernel gets all the registers as parameters whereas the ones that are not updated are stated as {\color{blue}{\tt \small const}}. The only difference between an inclusion and the corresponding exclusion kernel for the same column is the sign of the update; i.e., $\mathrel{-}=$ is used instead of $\mathrel{+}=$. 

\begin{center}
\begin{tabular}{c}
\begin{lstlisting}[caption={\small Reduction over the $\mathbf{x}$ array.}, numbers=left, linewidth=8.3cm, label={lst:prodReduce}]
#define C double
__device__ __inline__ void prodReduce(C& product, const C& reg0, const C& reg1, const C& reg2, const C& reg3, const C& reg4, const C& reg5) 
{
    product *= reg0;
    product *= reg1;
    product *= reg2;
    product *= reg3;
    product *= reg4;
    product *= reg5;
}
\end{lstlisting}
\end{tabular}
\end{center}

For all the kernels generated, the compiler can produce the PTX with Register Allocation optimization, as the exact registers to read/write can be determined at compile time. Although this is not necessary to store $\mathbf{x}$ on registers, the matrix entries can also be injected into the kernels as in Listing~\ref{lst:colKernel}, which removes the need to make access for $\mathbf{A}$ during execution. Once the updates on $\mathbf{x}$ are done, each kernel performs a product reduction over the elements of $\mathbf{x}$. This corresponds to lines~\ref{ln:reduce1}–\ref{ln:reduce2} in Alg.~\ref{alg:sparyser}. While generating {\tt \small prodReduce}, given in Listing~\ref{lst:prodReduce}, we manually unrolled the loop. However, this could be done by the CUDA compiler with $n$ being constant.

There may be three potential problems of this approach while leveraging GPU-based parallelism:
\begin{itemize}[leftmargin=*]
\item A sparse matrix usually has an irregular nonzero distribution in which the rows/columns usually have different numbers of nonzeros. Hence, the row-sum updates in between lines~\ref{ln:sumfor}--\ref{ln:xupdate} of Alg.~\ref{alg:sparyser} may incur a load imbalance among the threads in a warp if they process, i.e., include/exclude, different columns since the number of {\bf for} loop iterations is equal to the number of nonzeros at column. This problem also exists in our generated kernels since the proposed techniques have not ({\em yet}) fine-tuned the assignment of iteration chunks to threads. 

\item The kernel generation has another, probably a more drastic problem; {\em control divergence}. Instead of executing the same instructions (albeit, a different number of times) in the baseline implementation, the threads now call different device kernels generated by our technique. Furthermore, even if the index of the changed bit, i.e., the ID $j$ of the included/excluded column, is the same for two threads, the direction of change, i.e., $s$, can be different in which case one thread can call the {\em inclusion} kernel whereas the other can call the {\em exclusion} one.\looseness=-1 

\item Although $n$, the matrix dimension is practically small enough to fit a thread private $\mathbf{x}$ array to 255 registers, the total number of registers is still limited to 65536 per SM. Considering the register usage also for other code/data, the proposed technique can drastically reduce the number of threads that can concurrently run on an SM. This can incur a significant performance bottleneck for kernel generation.\looseness=-1  
\end{itemize}

The next section focuses on the first two problems whereas Section~\ref{sec:hybrid} focuses on the last one. 

\section{Minimizing Control Divergence}\label{sec:diverge}
The proposed kernel generation technique generates $2 \times (n - 1)$ kernels and makes each thread call one for the row-sum update in a single iteration. With {\em Single Instruction, Multiple Threads}~(SIMT), for efficiency, the threads within a warp need to run the same instruction in parallel. 
Otherwise, when they run different kernels, the instruction within those kernels differs, which leads to {\em control divergence}. As described before, the kernel to run depends on the changed bit $j$ in the Gray code and the change direction $s$. Assuming the matrix is $5 \times 5$ there will be $2^{5-1} - 1 = 15$ iterations and the changed-bit sequence is [0, 1, 0, 2, 0, 1, 0, 3, 0, 1, 0, 2, 0, 1, 0]. Note that for each column ID, the odd indexed appearances, e.g., 1st, 3rd, 5th etc., correspond to a column inclusion and the even indexed appearances correspond to an exclusion. With 5 threads and hence chunks containing 3 iterations, the kernel call pattern is  
\begin{center}
\begin{small}
\begin{tabular}{c|r|r|r|r|r|c}
& \multicolumn{5}{c}{Thread ID}& \\
Local  &  &  & &  &  & \#Distinct, \\
Iteration & 0 & 1 & 2& 3 & 4 & Concurrent Kernels \\\hline
0 & +0 & +2 & -0 & +1 & +0 & 4 \\
1 & +1 & +0 & +3 & -0 & -1 & 5 \\
2 & -0 & -1 & +0 & -2 & -0 & 4
\end{tabular}
\end{small}
\end{center}
which shows that 4, 5, and 4 distinct kernels, i.e., signed IDs, are executed in each iteration, respectively. For the table above, the sign +/- of the entries indicates if the flipped bit is changed from 0 to 1~(inclusion) or 1 to 0~(exclusion), respectively. This divergence forces the control unit to serialize the execution. 

Let $\tau$ be the number of threads. 
In a straightforward iteration space distribution, e.g.,~\cite{kaya19}, one can set the chunk size to $\lceil{(2^{n-1} - 1)/\tau}\rceil$ which will probably yield thread groups with divergent execution flow. To address this problem, here we exploit the recursive nature of the Gray codes and the recursive construction steps; {\em reverse}, {\em concatenate} and {\em prefix}. An \(n\)-bit Gray code can be constructed from two \((n-1)\)-bit Gray codes by (1) reversing the second one, (2) concatenating the first with the reversed form of the second, and (3) using a prefix bit {\bf 0} for the first, and {\bf 1} for the second. As an example, following these steps, a 3-bit Gray code is constructed in Figure~\ref{fig:gray_code_construction}. 

\begin{figure}[htbp]
    \centering
    \includegraphics[width=0.43\textwidth]{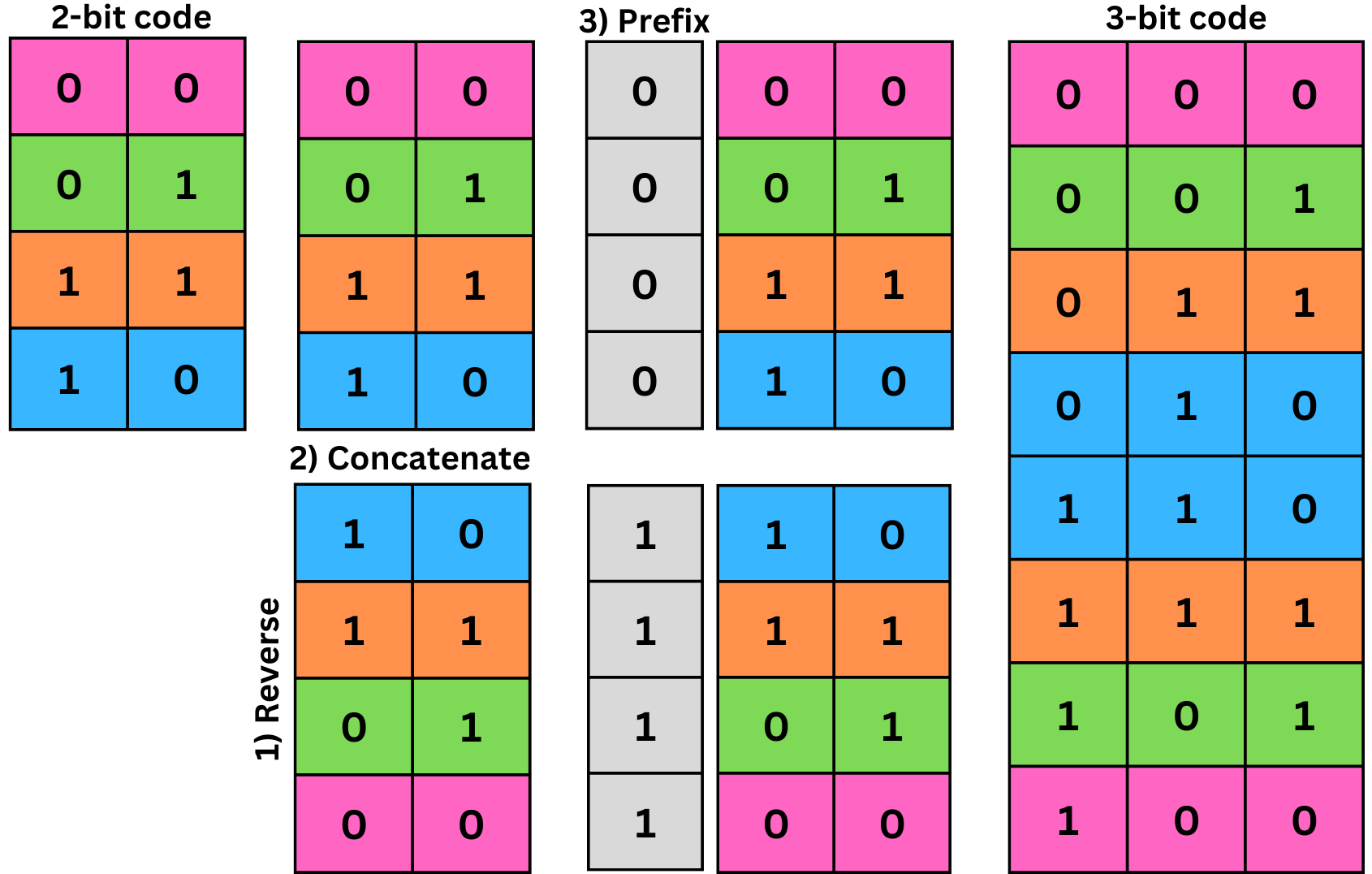}
    \caption{\small 3-bit Gray code construction from 2-bit sequences using the reverse, concatenate, and prefix method.}
    \label{fig:gray_code_construction}
\end{figure}

Since a Gray code can be constructed recursively, its {\em signed} changed-bit sequence~({\sc Scbs}) can also be described recursively. Let {\sc Scbs}(\(k\)) be the signed changed-bit sequence of an \(k\)-bit Gray code. Its recursive construction is given by: 

\[
\textsc{Scbs}(k) = [\textsc{Scbs}(k-1), +(k-1), -\textsc{Scbs}(k-1)^{R}]
\]
where $\textsc{Scbs}(k-1)^R$ denotes reversed $\textsc{Scbs}(k-1)$ and $-\textsc{Scbs}(k-1)$ denotes $\textsc{Scbs}(k-1)$ with alternated signs. For instance, having $\textsc{Scbs}(2) = [+0, +1, -0]$   
\begin{align*}
\textsc{Scbs}(3) &= [\textsc{Scbs}(2), +2, -\textsc{Scbs}(2)^{R}]\\
&= [[+0, +1, -0], +2, [+0, -1, -0]]
\end{align*}
and similarly $\textsc{Scbs}(4) = [\textsc{Scbs}(3), +3, -\textsc{Scbs}(3)^{R}]$ is 
\begin{align*}
[&[+0, +1, -0, +2, +0, -1, -0], +3, \\
 &[+0, +1, -0, -2, +0, -1, -0]]
\end{align*}
Note that for an $n \times n$ matrix, we are interested in $\textsc{Scbs}(n-1)$ for $(n-1)$-bit Gray codes which contain $2^{n-1} - 1$ elements. Based on the recursive construction of Gray codes, we have the following theorem:
\theoremstyle{plain}
\begin{theorem} For $1 \leq i \leq 2^{n-1} - 1$, the $i$th entry of $\textsc{Scbs}(n-1)$, has the changed bit ID $j$ if $2^j$ is the largest power of 2 dividing $i$. Furthermore, the sign of the same entry is {\bf +}~(indicating 0$\rightarrow$1) when $$(i - 2^j) / 2^{j+1}$$ is even, and {\bf -}~(indicating 1$\rightarrow$0) otherwise. 
\begin{proof}
For an $(n-1)$-bit Gray code, the $j$th bit changes for the first time at index $i = 2^j$. Furthermore, the sign of this change is {\bf +}~(see the recursive construction above). After this location, $j$ will continue to appear in $\textsc{Scbs}(n-1)$ in every other $2^{j+1}$th entry. Hence, $j$ will appear at the $i$th location if $2^j$ is the largest power of $2$ dividing $i$. For each appearance, the sign will alternate and the parity of $(i - 2^j) / 2^{j+1}$ is sufficient to find the sign. 
\end{proof}
\label{thm:scbs}
\end{theorem}

\begin{lem} Assume that the chunk size is a power of two, i.e., $2^k$ for a given $k$, instead of $\lfloor{2^{n-1}/\tau}\rfloor$ as suggested by the literature. Then, there is no control divergence among the threads except for two iterations. 
\label{lem:lem}
\end{lem}

\begin{proof} 
Consider the global index of the $\ell$th local iteration of each chunk for $1 \leq \ell \leq 2^k$. With SIMT, each $t$th thread, $0 \leq t < \tau$, will call the kernel based on $\textsc{Scbs}(n-1)[\ell + t \times 2^k]$. Let $i_t = \ell + t \times 2^k$ There are three cases:
\begin{itemize}
\item $\ell = 2^{k-1}$: the largest power of 2 dividing $i_t$ is $2^{k-1}$ for all $t$. However, the parity of $(i_t - 2^{k-1}) / 2^k$ equals the parity of $t$. Based on Theorem~\ref{thm:scbs}, in this case, all threads will call a kernel for column $k-1$. However, half of them will call an exclusion kernel, and the other half will call the inclusion version. Hence, there is a control divergence.  
\item $\ell = 2^{k}$: the largest power of 2 dividing $i_t$ is $2^k$ for even $t$, and $2^{k+1}$ for odd $t$. Based on Theorem~\ref{thm:scbs}, half of the threads call a kernel for column $k$, and the other half will call one for column $k + 1$. Hence, there is a control divergence.  
\item $\ell \neq 2^{k-1}$ and $\ell \neq 2^k$: the largest power of 2 dividing $i_t$ equals the largest power of 2 dividing $\ell$. Let $j$ be that power, i.e., the ID of the column which will be updated for row sums. The parity of $(i_t - 2^{j}) / 2^{j + 1}$ equals to the parity of $(\ell - 2^{j}) / 2^{j + 1}$. Based on Theorem~\ref{thm:scbs}, all the threads will call the same kernel and there will be no control divergence. 
\end{itemize}
\end{proof}

As a side note, for $\ell = 2^{k-1}$ in the proof of Lemma~\ref{lem:lem}, the changed-bit ID, $k-1$, is the same. Hence only two, the inclusion and exclusion kernels of column $k-1$, will be called. In this case, the cost of divergence will not be much. For the later iteration, $i = 2^{k}$, four different inclusion/exclusion kernels can be called for columns $k$ and $k+1$. In addition, having no control divergence implies that the thread loads are perfectly balanced since the same column is processed in concurrent iterations. 

Let $\tau$ be the number of threads suggested by the CUDA Occupancy API \cite{NVIDIA} and $\rho(x)$ is a function that returns the largest power of 2 less than or equal to a given integer \(x\). Following Lemma~\ref{lem:lem}, the next approach we propose is setting $\Delta = \rho(\lfloor2^{n-1} / \tau\rfloor)$ for the first kernel launch and use $\tau$ threads. Then the remaining $(2^{n-1} - 1) - (\Delta \times \tau)$ can also be handled similarly. This is repeated until all the iterations are consumed. Algorithm~\ref{alg:kernel_coalescing} summarizes our approach with a practical heuristic that limits the minimum chunk size to 1024 for each call. The last kernel launch at line~\ref{ln:lastcall} with parameters ($start$, $\Delta$, $end$) initiates a chunk size, $\Delta = 1024$, which is larger than required to avoid divergence. Some later threads do not contribute to the permanent computation for this last launch, i.e., the number of working threads will be smaller than $\tau$.\looseness=-1
 
\begin{algorithm}[H]
\setstretch{1.1}
\caption{: {\sc GenerateLaunchParameters}} \label{alg:kernel_coalescing}
    \small
\begin{flushleft}
    \small
    \textbf{Input:} \hspace*{1.5ex}$\tau$: Number of threads \\ \hspace*{8.7ex}$n$: Dimension of the matrix $\Ab$ \\
     \textbf{Output:} $\mathcal{K}$: Set of parameters for main kernel launches  
\end{flushleft}
\begin{algorithmic}[1]
\algrule
    \State $\mathcal{K} \gets \varnothing$
    \State $start \gets 1$
    \State $end \gets 2^{n - 1}$

    \While{$end - start > 0$} \algorithmiccomment{there exist more iterations}
        \State $\Delta \gets 1024$ 
        \While{$(\Delta \times \tau) \leq end - start$}
            \State $\Delta \gets \Delta \times 2$
        \EndWhile
        \State $\Delta \gets \Delta / 2$\algorithmiccomment{chunk size}

        \If{$\Delta = 512$}
        \State {$\mathcal{K} \gets \mathcal{K}$ $\cup$ \{($start$, $1024$, $end$)\}}\label{ln:lastcall}
            \State \textbf{break}
        \EndIf

        \State {$\mathcal{K} \gets \mathcal{K}$ $\cup$ \{($start$, $\Delta$, $end$)\}}

        \State $start \gets start + \tau \times \Delta$
    \EndWhile

\end{algorithmic}
\end{algorithm}

\section{Code Generation with Hybrid Memory Usage}\label{sec:hybrid}

Using a register for each row sum in $\xb$ can limit the occupancy of the GPU. To mitigate this problem, we propose  {\em hybrid-memory code generation} that uses global memory to store $\xb$ entries in addition to registers. We first prove the following lemma on the change bit pattern frequency for a Gray code.

\theoremstyle{plain}
\begin{lem}
For an $(n-1)$-bit Gray code, the probability of an entry in {\sc Scbs}$(n-1)$ containing the value $0 \leq j < n-1$, positive or negative, is 
\begin{equation}
2^{n-j-2} \Big/ (2^{n-1} - 1).
\end{equation}
\label{eq:prob}
\end{lem}
\begin{proof}
Based on Theorem~\ref{thm:scbs}, the value $0 \leq j < n-1$ first appears at $2^j$th location, and then periodically reappears in every $2^{j+1}$th entry. Hence, the number of appearances of $j$ is exactly 
\begin{align*} 
\left\lfloor 1 + \frac{(2^{n-1} - 1) - 2^j}{2^{j+1}}\right\rfloor &= \left\lfloor \frac{(2^{n-1} - 1) + 2^j}{2^{j+1}}\right\rfloor \\
&= \left\lfloor 2^{n-j-2} + \frac{1}
{2} - \frac{1}{2^{j+1}}\right\rfloor \\
&= 2^{n-j-2}.
\end{align*}
Since there are $2^{n-1} - 1$ entries in {\sc Scbs}($n$), \eqref{eq:prob} holds. 
\end{proof}

For simplicity, in the rest of the paper, we will use $1 / 2^{j+1}$ for the probability of an entry being $j$, which is close enough for our purposes. To include/exclude a column $j$ in an iteration, consider the accesses to the $\xb$ vector: for a dense matrix $\Ab$, all $\xb$ entries are updated. When the matrix is sparse, this is not the case. For instance, for the matrix in Fig.~\ref{fig:gray_code_change}, \(\mathbf{x}\)[2] has approximately 1/2 + 1/16 + 1/32 = 19/32 probability of being modified in an iteration whereas the same probability for \(\mathbf{x}\)[4] is 1/4 + 1/8 = 3/8.\looseness=-1

Following Lemma~\ref{eq:prob}, for an integer $1 \leq c \leq n-1$, the probability of an absolute value of entry in {\sc Scbs}($n-1$) being smaller than $c$ approximately equals to 
$$\sum_{j = 0}^{c-1}\frac{1}{2^{j+1}} = 1 - \frac{1}{2^{c}}.$$ 
Let ${\cal I}_c$ be the set of row ids such that $i \in {\cal I}_c$ if  $a_{i,j} \neq 0$ for at least one $0 \leq j < c$. 
When only the $\xb[i]$s for all $i \in {\cal I}_c$ are stored in registers and the rest are kept in global memory, for $1 - 1/2^{c}$ of the iterations there will not be any global memory access. In this case, only $|{\cal I}_c|$, instead of $n$, registers per thread will be used, allowing the main kernel to be launched with more threads. To increase the impact of this optimization and make the matrix suitable for it, we propose {\em permanent ordering} that organizes the rows/columns to obtain the sparsity pattern given in Fig.~\ref{fig:permanent_ordering_pattern}. Note that {\sc{perm}}$(\Ab) =$ {\sc{perm}}$(\mathbf{P}\Ab\mathbf{Q})$, 
where $\mathbf{P}$ and $\mathbf{Q}$ are any permutation matrices. That is the permanent is preserved under any row or column reordering.

\begin{figure}[htbp]
     \centering
     \begin{subfigure}[b]{0.45\linewidth}
         \centering
         \includegraphics[width=\linewidth]
         {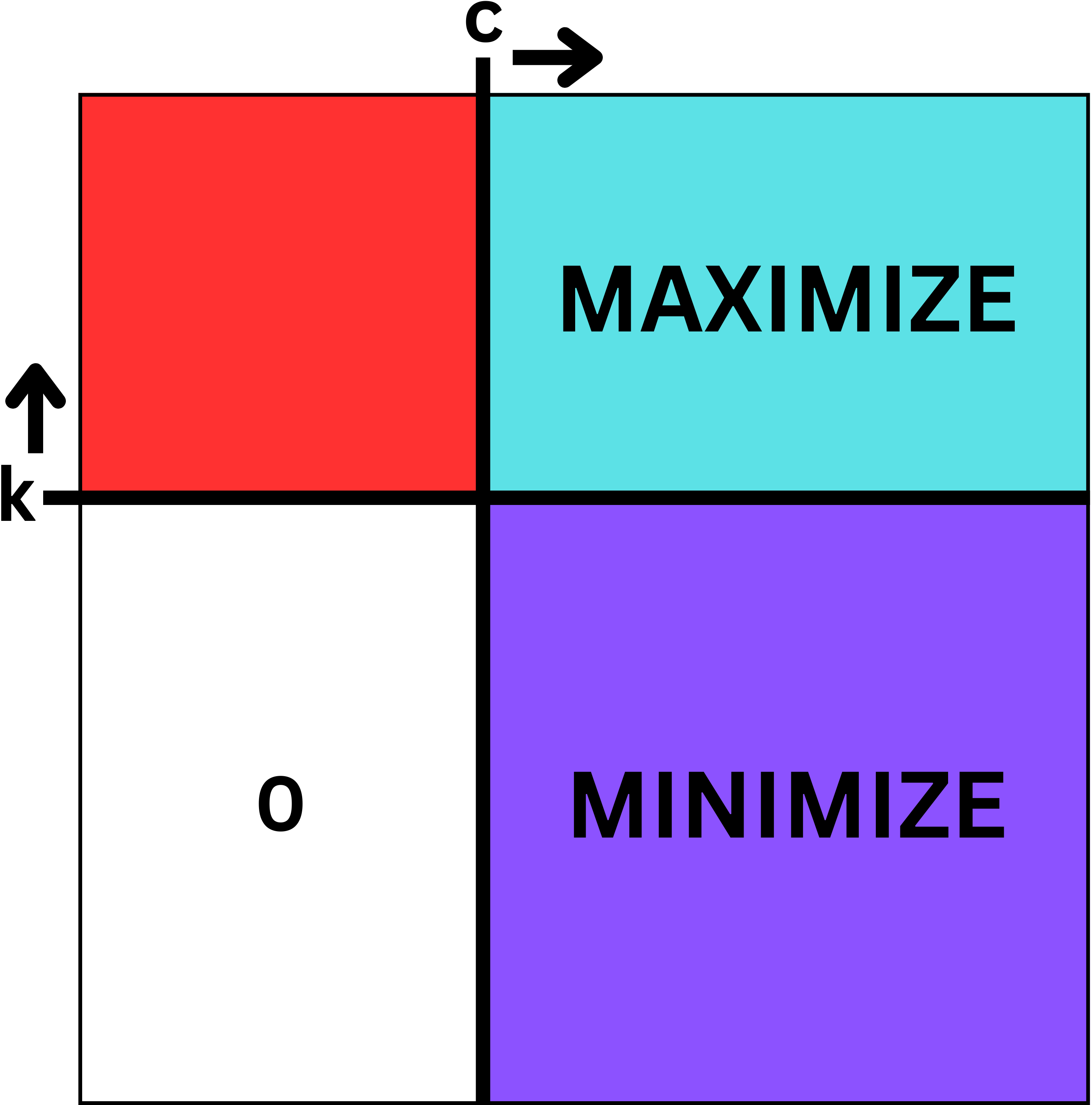}
         \caption{\small The desired nonzero pattern from ordering. The number of rows, $k$, and columns, $c$, are marked with the desired direction to highlight the variables influencing quality.
         }
         \label{fig:permanent_ordering_pattern}
     \end{subfigure}
     \hfill
     \begin{subfigure}[b]{0.52\linewidth}
         \centering
         \includegraphics[width=\linewidth]
         {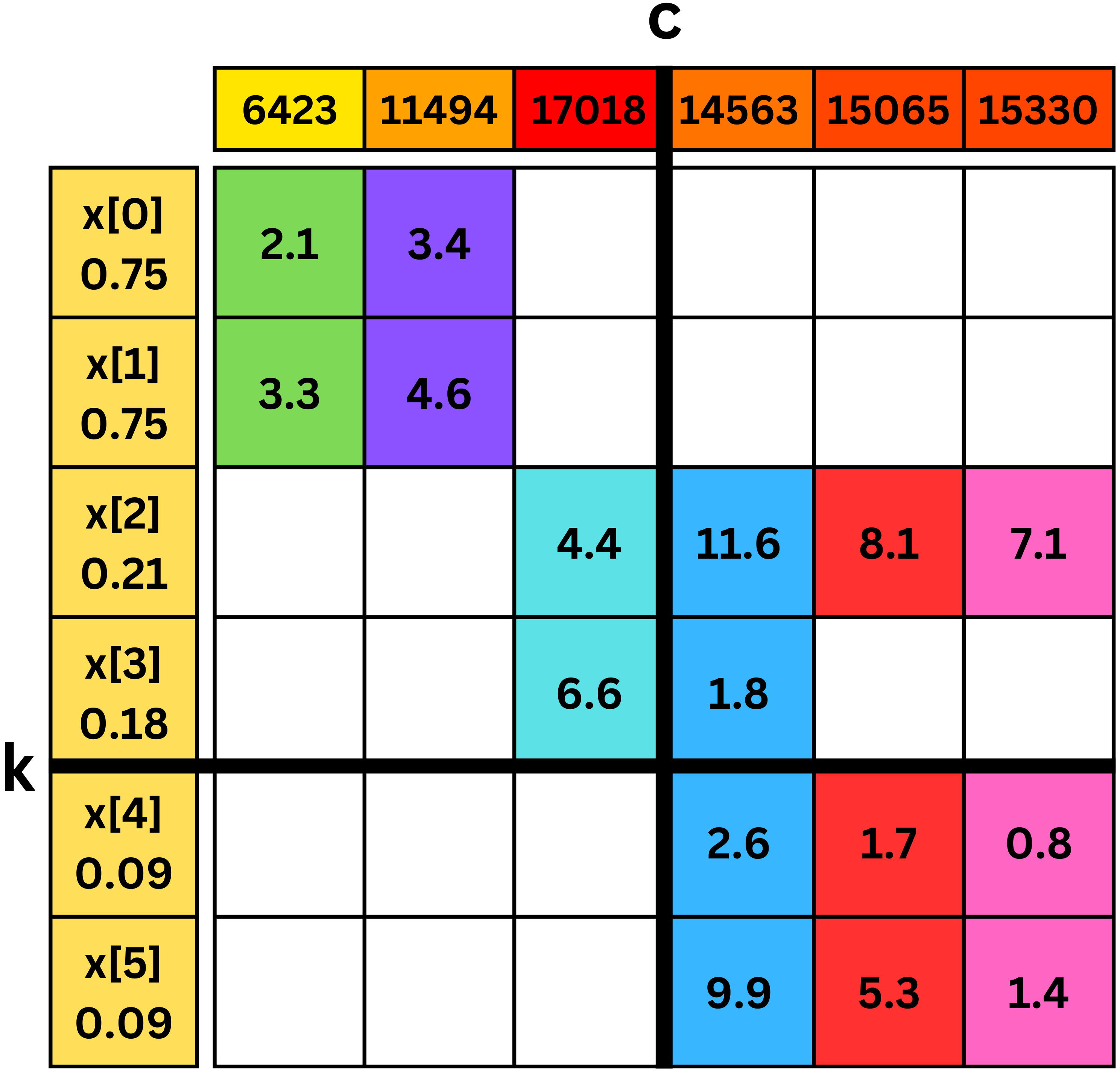}
         \caption{\small The $6 \times 6$ sparse matrix in Fig.~\ref{fig:gray_code_change} ordered. The row/column partitioning for register/global memory usage is also given along with the computed score for each additional column. }  
         \label{fig:permanent_ordering_matrix}
     \end{subfigure}
     \caption{\small The desired sparsity pattern from permanent ordering~(left). The matrix from Fig.~\ref{fig:gray_code_change} ordered via permanent ordering~(right).} 
     \label{fig:permanent_ordering}
\end{figure}

Assume that the rows in the matrix are ordered in a way that the ones in ${\cal I}_c$ appear first. As Fig.~\ref{fig:permanent_ordering_pattern} shows, this strategy partitions the matrix into 4 areas: (1) the top-left~(red) area contains the nonzeros of $c$ columns whose inclusion/exclusion kernels will only touch registers, (2) the bottom-left~(white) area with no nonzeros, (3) the top-right~(blue) area which contains the register-updating nonzeros of the remaining $(n-c)$ columns, and (4) the bottom-right~(purple) area containing the global-memory-updating nonzeros of the same set of columns. To maximize the number of iterations not performing global memory accesses, i.e., the number of inclusion/exclusion kernel calls that only update registers, we want $c$ to be larger. To use less registers per thread, we want $k$ to be smaller. Nevertheless, these two goals conflict with each other. It is also preferable that the top-right area have more nonzeros than the bottom-right one for a small number of global memory updates.

\begin{algorithm}
\setstretch{1.1}
\caption{: {\sc PermanentOrdering}} 
\label{alg:permanent_ordering}
    \small
\begin{flushleft}
\small
\textbf{Input:} \hspace*{1.3ex}($rptrs,  cids, rvals$) \algorithmiccomment{\small CSR of $n \times n$ sparse matrix $\Ab$} \\ \hspace*{8.4ex}($cptrs, rids, cvals$) \algorithmiccomment{CSC of $n \times n$ sparse matrix $\Ab$}\\
\textbf{Output:} {\tt rowPerm} \algorithmiccomment{\small Row permutation for $\Ab$} \\ \hspace*{8.3ex}{\tt colPerm} \algorithmiccomment{\small Column permutation for $\Ab$} 
\end{flushleft}
\algrule
\begin{algorithmic}[1]

\For{$j \gets 0$ to $n-1$}
    \State {$cdeg[j] \leftarrow cptrs[j+1] - cptrs[j]$ }
\EndFor
\For{$i \gets 0$ to $n-1$}
\State {$rmark[i] \leftarrow {\mathbf{false}}$}
\EndFor
\algrule

\State{$ridx \leftarrow 0$} \algorithmiccomment{\small No of permuted (moved up) rows}
\For{$cidx \gets 0$ to $n-1$} 

    \hrulefill
   \State{$\blacktriangleright$ \footnotesize Permute the $col$umn with min. no of unpermuted nonzeros}
    \State{$col \gets \arg\min_{j} \{cdeg[j]$\}}
    \State{$colPerm[cidx] \gets col$}
    \State $cdeg[col] \gets \infty$

    \For{each $row$ s.t. $a_{row, col} \neq 0$} \algorithmiccomment{\small via CSC}
    
        \hrulefill
   \State{$\blacktriangleright$ \footnotesize Permute each unpermuted $row$ sharing a nonzero with $col$}
        \If{$rmark[row] = \mathbf{false}$} 
            \State {$rmark[row] \gets \mathbf{true}$}
            \State{$rowPerm[ridx] \gets row$}
            \State{$ridx \gets ridx + 1$}
            
\hrulefill
   \State{$\blacktriangleright$ \footnotesize Update the no of unpermuted nonzeros for each $col'$}
            \For{each $col'$ s.t. $a_{row, col'} \neq 0$} \algorithmiccomment{\small via CSR}
                \State{$cdeg[col'] \leftarrow cdeg[col'] - 1$}
            \EndFor
        \EndIf
    \EndFor
\EndFor

\end{algorithmic}
\end{algorithm}

\begin{algorithm}
\setstretch{1.1}
\caption{: {\sc Partitioning}} 
\label{alg:partitioning}
    \small
\begin{flushleft}
\small
\textbf{Input:} 
\hspace*{1.3ex}($cptrs, rids, cvals$) \algorithmiccomment{ CSC of $n \times n$ ordered matrix $\Ab'$}\\
\hspace*{8.4ex}$GRratio = 16$ \algorithmiccomment{Rel. cost of glob./reg. mem. updates }\\
\textbf{Output:} $c$ \algorithmiccomment{\small Number of columns only updating the registers.} \\ \hspace*{8.3ex}$k$ \algorithmiccomment{\small Number of entries of $\xb$ to be stored in registers. } 
\end{flushleft}
\algrule
\begin{algorithmic}[1]
    \State $k$ $\gets 0$
    \State $c$ $\gets 0$
    \State $bestScore$ $\gets 0$
    \State $nrows$ $\gets 0$
    \For{$j \gets 0$ \textbf{to} $n - 1$}
        \State{$\blacktriangleright$ \footnotesize Number of rows touched by the first $j+1$ columns (via CSC)} 

        \State $nrows$ $\gets$ max(\texttt{$nrows$,$(\max\{i: a_{i,j} \neq 0\} + 1)$})\label{ln:nrows} 
        \State $nregisters$ $\gets$ $nrows \times 2$ \label{ln:nregs} \algorithmiccomment{\small double precision storage} 
        \State $regCost$ $\gets {nregisters \times \left(1 - 2^{-{(j+1)}}\right)}$~\label{ln:regcost} 
        \State $globCost$ $\gets {(n - nrows) \times 2^{-(j +1)}}\times $ $GRratio$~\label{ln:globcost}

        \State{$\blacktriangleright$ \footnotesize Estimate \# of threads to launch based on register usage} 
        \State $\tau$ $\gets$ {\sc CalculateNoThreads}($nregisters$)~\label{ln:tau} 
        
        \State $currentScore$ $\gets \tau / (regCost + globCost)$
        \If{($currentScore$ $>$ $bestScore$) \textbf{or} ($nrows$ = $k$)} 
            \State $bestScore$ $\gets$ $currentScore$ 
            \State $k$ $\gets$ $nrows$ 
            \State $c$ $\gets j+1$ 
        \EndIf
    \EndFor
\end{algorithmic}
\end{algorithm}

The proposed ordering algorithm is described in Alg.~\ref{alg:permanent_ordering}. The algorithm iterates $n$ times and at each iteration, it chooses a column $col$ with the minimum number of nonzeros on unordered rows. Along with $col$, the unordered rows sharing a nonzero with $col$ are also ordered. For each such $row$, the unordered nonzero counts are updated for the columns sharing a nonzero with $row$. The matrix from Fig.~\ref{fig:gray_code_change} ordered via this approach is given in Figure~\ref{fig:permanent_ordering_matrix}. After ordering, we set the values for $k$ and $c$ with a partitioning technique given in Algorithm~\ref{alg:partitioning}. Starting from the first one, for each $0 \leq j < n$, Alg.~\ref{alg:partitioning} checks if column $j$ should be added to the pure register area.\looseness=-1

Let $nrows$ be the number of rows~(line~\ref{ln:nrows}) touched by the first $j+1$ columns. Algorithm~\ref{alg:partitioning} first computes the worst-case cost of register updates as $regCost = nregisters \times (1 - 2^{-(j+1)})$ since for approximately $(1 - 2^{-(j+1)})$ of the iterations, $nrows$ number of double-precision $\xb$ entries, stored in 32-bit registers will be updated~(line \ref{ln:regcost}) which is the amount of computation power the partitioning can obtain per computation. Similarly, the global memory update cost, $globCost$, is computed by considering the $(n-nrows)$ values stored in global memory which will be touched in $2^{-(j+1)}$ of the iterations. $globCost$ is then scaled by the relative cost of a global memory update to a register update~(line \ref{ln:globcost}). Based on our preliminary benchmarks, we used $GRratio = 16$ as the scaling factor. Let $\tau$ represent the estimated number of threads that can be launched given each thread uses $nregisters$~(line \ref{ln:tau}). The $currentScore$ of adding the $j$th column, i.e., setting $c = j+1$, is then calculated as $\tau / (regCost + globCost)$. If $currentScore$ is bigger than the $bestScore$, the heuristic adds column $j$ into the pure register area by setting $k = nrows$ and $c = j+1$.
Consider the matrix in Fig.~\ref{fig:permanent_ordering_matrix} partitioned into four regions with $c = 3$ and $k = 4$. This configuration makes the first four most {\em valuable} elements of \(\mathbf{x}\) stored in registers while the last two are in global memory. The last three columns are not included in the register area since their scores, given on top of the columns in Fig.~\ref{fig:permanent_ordering_matrix}, are not larger than the one computed for the third column.

Listing~\ref{lst:colKernel} is the inclusion kernel for column 0 of the original matrix given in Fig.~\ref{fig:gray_code_change}. This column is now column 3 in the ordered matrix given in Figure~\ref{fig:permanent_ordering_matrix}. The new kernel, {\tt \small hybrid\_c3\_inc}, which is using both the registers and the global memory, is given in Listing~\ref{lst:hybridColKernel}. In addition to the parameters for 4 registers for the first 4 row sums, each inclusion/exclusion kernel now takes the $\xb$ pointer for the global memory part as input. The indexing mechanics, i.e., $nthreads \times row + tid$, for the global part of $\xb$ is designed to make the accesses coalesced. Finally, the variables used for indexing, i.e., $nthreads$ and $tid$, are declared as {\tt \small \color{blue}{volatile}} to avoid the compiler putting the $\xb[.]$ indexes to registers which drastically limits the number of threads that can be launched. The global part of the final product is computed within this kernel.

The kernel {\tt \small hybridProdReduce} in Listing~\ref{lst:hybridProdReduce} adds the registers' contribution to the final product. Note that for register-updating columns, i.e., columns $j < c$, the {\tt \small {globalProduct}} is not changed and can be reused. Hence, when $j < c$, the global memory is not accessed even for reading. 

\begin{center}
\begin{tabular}{c}
\begin{lstlisting}[caption={\small Hybrid inclusion kernel for column 3.}, numbers=left, linewidth=8cm, label={lst:hybridColKernel}]
#define C double
__device__ __inline__ void hybrid_c3_inc(
    C& product, C& globalProduct 
    const C& reg0, const C& reg1, 
    C& reg2, C& reg3
    C* x, const volatile unsigned& nthreads,
    const volatile unsigned& tid) {  
    
    reg2 += 11.600000;
    reg3 += 1.800000;

    globalProduct = 1;
    x[nthreads * 0 + tid] += 2.600000;
    globalProduct *= x[nthreads * 0 + tid];
    x[nthreads * 1 + tid] += 9.900000;
    globalProduct *= x[nthreads * 1 + tid];

    hybridProdReduce(product, globalProduct, 
                        reg0, reg1, reg2, 
                        reg3);
}
\end{lstlisting}
\end{tabular}
\end{center}

\begin{center}
\begin{tabular}{c}
\begin{lstlisting}[caption={\small Hybrid reduction over the $\mathbf{x}$ array.}, numbers=left, linewidth=8cm, label={lst:hybridProdReduce}]
#define C double
__device__ __inline__ void hybridProdReduce(
    C& product, const C& globalProduct, 
    const C& reg0, const C& reg1, 
    const C& reg2, const C& reg3) {
    product *= reg0;
    product *= reg1;
    product *= reg2;
    product *= reg3;
    product *= globalProduct;
}
\end{lstlisting}
\end{tabular}
\end{center}

\section{Experimental Results}

We first describe the components in the experimental setting. After these, we will present the performance results.

\subsection{Architectures }
We conducted our experiments on three servers. The first two, \textbf{Arch-1} and \textbf{Arch-2}, are used for running the current state-of-the-art CPU implementation from the literature. \textbf{Arch-1} has a theoretical performance of 7 TFLOPs and is equipped with two 56-core Intel Xeon Platinum 8480+ 2.0 GHz processors, each costing \$10,710, and 256 GB of DDR5 memory. \textbf{Arch-2} offers a theoretical performance of 3.2 TFLOPs, with two Intel Xeon 6258R 2.70 GHz processors, providing 56 cores in total and 192 GB of memory. \textbf{Arch-3}, on the other hand, was dedicated to running the GPU implementations mentioned and proposed in this paper. It is equipped with an Nvidia A100 GPU with 80 GB of device memory, having a theoretical performance of 9.6 TFLOPs which costs almost the same as the CPUs in \textbf{Arch-1}.

\subsection{Algorithms and Implementations}\label{sec:algos}
There are four implementations tested in this study:

\begin{enumerate}[leftmargin=*]
\item{\textit{CPU-SparsePerman}}: As the CPU baseline, we directly utilized {\textsc{SparsePerman}} presented in Algorithm \ref{alg:sparyser}. In addition, we introduced two key improvements to the algorithm as suggested by the literature~\cite{kaya19}. First, we ordered the matrix using degree-sort in ascending order~\cite{kaya19}. As shown in Lemma~\ref{eq:prob}, the probability of an entry in CBS($n$) equal to $j$ is high when $j$ is small. By sorting the columns in ascending order based on the number of nonzeros they have, the frequently accessed columns perform fewer updates, reducing the amount of data movement of the \textbf{for} loop at line~\ref{ln:sumfor}. Second, we tracked the zeros of $\xb$. In the presence of a zero, all the expensive multiplications and the update on the result in Algorithm~\ref{alg:sparyser} are skipped. 
\item{\textit{GPU-SparsePerman}}: In addition to the CPU, we ported {\textsc{SparsePerman}} efficiently to the GPU with the implementation explained in Section~\ref{subsec:parper}, and used in the beginning of Section~\ref{sec:gpuper}. The vector $\xb$ and the matrix \textbf{A} are stored in shared memory to increase the arithmetic intensity of the implementation. The same improvements applied to the CPU version were incorporated into the GPU algorithm. 
\item{\textit{CodeGen-PureReg}}: This is our kernel-generation-based implementation that only uses registers as explained in Section~\ref{sec:gpuper} including the technique used for minimizing control divergence in Section~\ref{sec:diverge}.
\item{\textit{CodeGen-Hybrid}}: This is the kernel-generation-based implementation that uses both registers and global memory. Algorithm~\ref{alg:permanent_ordering} and Algorithm~\ref{alg:partitioning} are used to order the original matrix and determine optimal values for $k$ and $c$.
\end{enumerate}

\begin{table}[ht]
\centering
\caption{\small Real-life matrices used in the experiments: the dimensions, total number of nonzeros, and the densities ($nnz / (n \times n)$).}\label{tab:real_life_set}
\scalebox{0.79}{
\begin{tabular}{l|rrr||l|rrr}
Matrix & \multicolumn{1}{c}{\(n\)} & \multicolumn{1}{c}{\(nnz\)} & density & Matrix & \multicolumn{1}{c}{\(n\)} & \multicolumn{1}{c}{\(nnz\)} & density\\
\midrule
bcsstk01       & 48 & 400 & 17.4\% & bcspwr02         & 49 & 167 & 7.0\% \\
mycielskian6     & 47 & 472 & 21.4\% & curtis54    & 54 & 291 & 10.0\% \\
mesh1e1        & 48 & 306  & 13.3\% & d\_ss  & 53 & 144 & 5.1\% 
\end{tabular}
}
\end{table}

\begin{table*}[htbp]
\caption{\small Execution times~(seconds) and the number of CPU/GPU threads used to compute the permanents of synthetic matrices across three different architectures.}\label{tab:synthetic_experiments}
\small
\centering
\scalebox{0.85}{
\begin{tabular}{l||rr|rr|rr|rr|rr}
\rowcolor{gray!20}
\multicolumn{11}{c}{$n = 40$} \\
\midrule
& \multicolumn{2}{c|}{\cellcolor{gray!20}$p = 0.1$} & \multicolumn{2}{c|}{\cellcolor{gray!20}$p = 0.2$} & \multicolumn{2}{c|}{\cellcolor{gray!20}$p = 0.3$} & \multicolumn{2}{c|}{\cellcolor{gray!20}$p = 0.4$} & \multicolumn{2}{c}{\cellcolor{gray!20}$p = 0.5$} \\
& Time & \# threads & Time & \# threads & Time & \# threads & Time & \# threads & Time & \# threads \\
\midrule
\textit{CPU-SparsePerman} (\textbf{Arch-1}) & 85.83 & 112 & 112.68 & 112 & 132.92 & 112 & 171.64 & 112 & 197.50 & 112 \\
\textit{CPU-SparsePerman} (\textbf{Arch-2}) & 362.51 & 56 & 449.17 & 56 & 483.46 & 56 & 660.82 & 56 & 468.83 & 56 \\
\textit{GPU-SparsePerman} \textbf{(Arch-3)} & 19.03 & 48384 & 28.91 & 41472 & 33.33 & 41472 & 37.38 & 41472 & 30.78 & 31104 \\ 
\textit{CodeGen-PureReg} \textbf{(Arch-3)} & 7.52 & 55296 & 7.22 & 55296 & 7.30 & 55296 & 7.95 & 55296 & 8.75 & 55296 \\ 
\textit{CodeGen-Hybrid} \textbf{(Arch-3)} & 3.18 & 138240 & 3.94 & 96768 & 4.77 & 82944 & 5.96 & 69120 & 6.51 & 69129 \\ 
\midrule
\rowcolor{gray!20}
\multicolumn{11}{c}{$n = 45$} \\
\midrule
& \multicolumn{2}{c|}{\cellcolor{gray!20}$p = 0.1$} & \multicolumn{2}{c|}{\cellcolor{gray!20}$p = 0.2$} & \multicolumn{2}{c|}{\cellcolor{gray!20}$p = 0.3$} & \multicolumn{2}{c|}{\cellcolor{gray!20}$p = 0.4$} & \multicolumn{2}{c}{\cellcolor{gray!20}$p = 0.5$} \\
& Time & \# threads & Time & \# threads & Time & \# threads & Time & \# threads & Time & \# threads \\
\midrule
\textit{CPU-SparsePerman} (\textbf{Arch-1}) & 2995.07 & 112 & 3478.29 & 112 & 4730.79 & 112 & 6430.01 & 112 & 7603.60 & 112 \\
\textit{CPU-SparsePerman} (\textbf{Arch-2}) & 13453.54 & 56 & 14537.87 & 56 & 17651.44 & 56 & 26871.01 & 56 & 28794.42 & 56 \\
\textit{GPU-SparsePerman} \textbf{(Arch-3)} & 777.07 & 41472 & 921.39 & 41472 & 1351.32 & 34560 & 1542.88 & 31104 & 1910.92 & 31104 \\ 
\textit{CodeGen-PureReg} \textbf{(Arch-3)} & 222.55 & 55296 & 239.19 & 55296 & 242.50 & 55296 & 269.96 & 55296 & 278.99 & 55296 \\ 
\textit{CodeGen-Hybrid} \textbf{(Arch-3)} & 92.06 & 124416 & 93.69 & 110592 & 155.28 & 69120 & 196.25 & 69120 & 251.07 & 55296 \\ 
\midrule
\rowcolor{gray!20}
\multicolumn{11}{c}{$n = 48$} \\
\midrule
& \multicolumn{2}{c|}{\cellcolor{gray!20}$p = 0.1$} & \multicolumn{2}{c|}{\cellcolor{gray!20}$p = 0.2$} & \multicolumn{2}{c|}{\cellcolor{gray!20}$p = 0.3$} & \multicolumn{2}{c|}{\cellcolor{gray!20}$p = 0.4$} & \multicolumn{2}{c}{\cellcolor{gray!20}$p = 0.5$} \\
& Time & \# threads & Time & \# threads & Time & \# threads & Time & \# threads & Time & \# threads \\
\midrule
\textit{CPU-SparsePerman} (\textbf{Arch-1}) & 27914.90 & 112 & 34747.61 & 112 & 42978.76 & 112 & 54653.47 & 112 & 63668.74 & 112 \\
\textit{GPU-SparsePerman} \textbf{(Arch-3)} & 6556.42 & 41472 & 9199.27 & 34560 & 11505.95 & 31104 & 14811.26 & 31104 & 16818.86 & 27648 \\ 
\textit{CodeGen-PureReg} \textbf{(Arch-3)} & 1860.66 & 55296 & 1943.09 & 55296 & 2154.58 & 55296 & 2214.10 & 55296 & 2382.50 & 55296 \\ 
\textit{CodeGen-Hybrid} \textbf{(Arch-3)} & 741.59 & 124416 & 1059.77 & 82944 & 1361.69 & 69120 & 1906.64 & 55296 & 2023.94 & 55296 \\ 
\end{tabular}
}
\end{table*}

\subsection{Datasets}

We utilized two sets of sparse matrices to test the implementations proposed in this paper. The first set contains synthetic matrices of sizes $n \in \{40, 45, 48\}$ and densities $p \in \{0.1, 0.2, 0.3, 0.4, 0.5\}$, generated using the Erdős–Rényi model. In this model, each $a_{i,j}$ is nonzero with $p$ probability. Therefore, any row/column is expected to have $p \times n$ nonzeros, and the total number of nonzeros in the matrix is expected to be $p \times n^2$. The values of these matrices are randomly chosen from the interval [$0,1$]. One caveat is that we rejected structurally rank-deficient matrices and continued the process until a matrix with a nonzero permanent was generated.

The second dataset consists of six real-life matrices from the Suite Sparse Matrix Collection~\cite{davis2011ufl}. We selected all the $n \times n$ square matrices with a nonzero permanent for $40 \leq n \leq 55$, given in Table~\ref{tab:real_life_set}. On these matrices, however, we tested only the most optimized implementation, \textit{CodeGen-Hybrid}~({\bf Arch-3}) and compared its performance to {\sc SparsePerman} executed on CPU with {\bf Arch-1} and on GPU with {\bf Arch-3}.

\subsection{Experiments on Synthetic Matrices}
Table~\ref{tab:synthetic_experiments} shows the results of the experiments on synthetic matrices. The first set of observations is as follows: (1) the proposed final implementation \textit{CodeGen-Hybrid} is $\approx31\times$ and $\approx118\times$ faster than \textit{CPU-SparsePerman} on {\bf Arch-1} and {\bf Arch-2}, respectively. (2) Furthermore, it is $8\times$ faster than \textit{GPU-SparsePerman} when they run on the same architecture {\bf Arch-3}. (3) With permanent ordering, partitioning, and using global memory in addition to the registers, \textit{CodeGen-Hybrid} is $1.74\times$ faster on average than \textit{CodeGen-PureReg}. 
 
\begin{figure*}[htbp]
    \centering
    \includegraphics[width=1\textwidth]{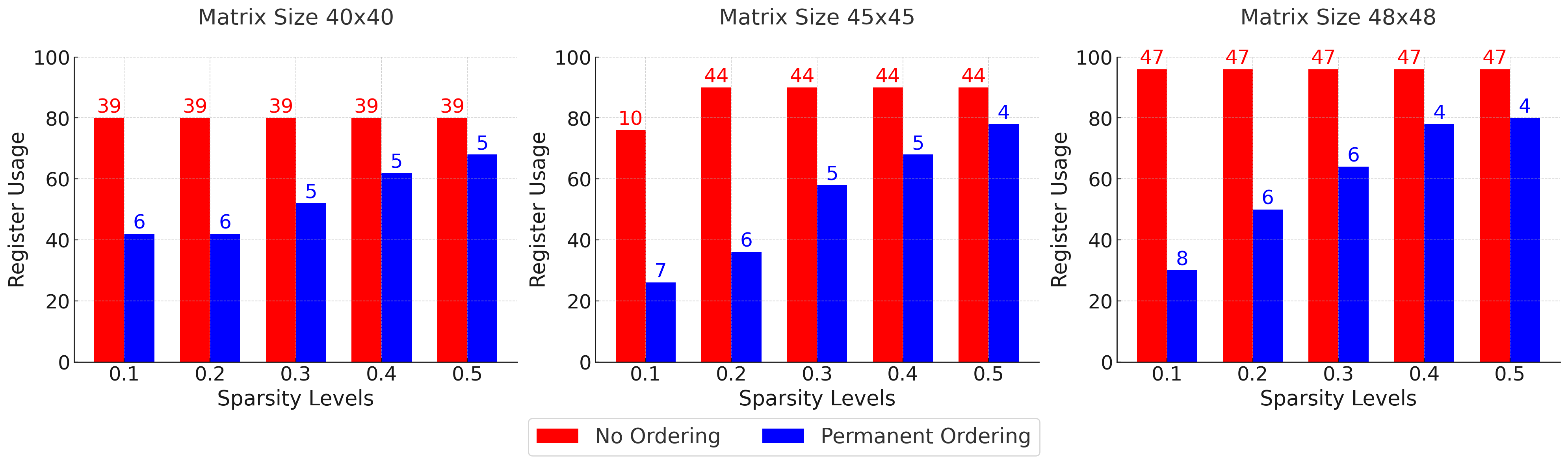}
    \caption{\small The number of registers used per thread~($y$-axes) to store $\xb$ values, i.e., $2 \times k$, for hybrid-memory
    kernel generation when permanent ordering is not applied~(red bars) and when it is applied~(blue bars) before partitioning the matrix with Algorithm~\ref{alg:partitioning}. The values over the bars show the number of columns performing only register updates, i.e., $c$, 
    also set by the partitioning algorithm.\looseness=-1}
    \label{fig:register_usage}
\end{figure*}

To analyze the impact of the permanent ordering, we removed it from the pipeline and used Algorithm~\ref{alg:partitioning} to partition the original matrix. Figure~\ref{fig:register_usage} shows the numbers of registers used per thread~(bars, $y$-axis), $2 \times k$, and columns with only register-updating kernels~(values on top of the bars), $c$, with and without ordering~(red and blue bars, respectively). The experiment shows that ordering significantly reduces register usage especially when the matrices are sparser, i.e., $p \leq 0.3$. However, even with ordering, the first few columns tend to have an entry among the last rows when the matrix gets denser. Hence, $k$ quickly becomes close to $n$. This is also reflected in performance; when $p$ increases Table~\ref{tab:synthetic_experiments} reports less speedup. For instance {\textit{CodeGen-Hybrid}} is $2.5\times$, $2.1\times$, $1.6\times$, $1.3\times$, and $1.2\times$ faster than {\textit{CodeGen-PureReg}}, respectively, for $p = 0.1, 0.2, 0.3, 0.4$, and $0.5$. This happens because the number of threads that can be launched with {\textit{CodeGen-Hybrid}} decreases when the density increases. However, for {\textit{CodeGen-PureReg}} the thread count is always the same at each table row since it depends only on $n$. 

We did not observe a meaningful decrease in speedups relative to the CPU-based implementations since they also get slower with density. For instance, the overall speedup of {\textit{CodeGen-Hybrid}} w.r.t. {\textit{CPU-SparsePerman}} on {\bf Arch-1} is $31.2\times$ whereas considering only the $p = 0.5$ experiments, the average speedup is $30.7\times$. However, although not drastically, the speedups relative to {\textit{GPU-SparsePerman}} decreases with increasing density; {\textit{CodeGen-Hybrid}} is $7.8\times$, $8.6\times$, $8.0\times$, $7.3\times$, and $6.9\times$ faster for $p = 0.1, 0.2, 0.3, 0.4$, and $0.5$, where the average speedup is $7.7\times$. The best speedup relative to {\textit{GPU-SparsePerman}} is obtained for $p = 0.2$. Lastly, there is no meaningful pattern in relative performances when $n$ increases. 

\subsection{Experiments on Real-Life Matrices}

Table~\ref{tab:real_life_experiments} shows the results of the experiments on real-life matrices. Overall, {\textit{CodeGen-Hybrid}} is $24.9\times$ and $4.9\times$ faster than {\textit{CPU-SparsePerman}} and {\textit{GPU-SparsePerman}}. In this experiment, the matrices can be classified into two with respect to the speedup obtained by {\textit{CodeGen-Hybrid}}; 
\begin{itemize} 
\item On $bcsstk01$, $mycielskian6$, and $mesh1e1$ the speedups are higher. Among these, $bcsstk01$ and $mesh1e1$ have mostly unique nonzero values. 
\item On the contrary, $curtis54$ and $bcspwr02$ are binary matrices, with all nonzero entries equal to $1$, and $d\_ss$ has more than half of the nonzeros as $1$ and $-1$. Furthermore, $d\_ss$ has only 40 unique nonzero values among 144.
\end{itemize}
Having the same values increases the probability of having a zero $\xb$ entry for many consecutive iterations. As stated in Section~\ref{sec:algos}, this is exploited by our CPU/GPU {\textit{SparsePerman}} implementations. For instance, the $prod$ computation in between lines 17--19 of Algorithm~\ref{alg:sparyser} and its contribution to $p$ at line 21 are skipped in $99.37\%$ of the iterations for $bcspwr02$, in $99.30\%$ of the iterations for $curtis54$, and in $99.21\%$ of the iterations for $d_{ss}$.
Note that even in the 0 case, the $\xb$ array is updated in both CPU/GPU {\textit{SparsePerman}} implementations. 

Even though {\textit{CodeGen-Hybrid}} does not exploit $0$s in $\xb$, it is still $2.1\times$ faster for real-life matrices. This is due to smart chunking and load distribution. There is almost no control divergence and load imbalance among the threads within a warp. In addition, the {\em achieved occupancy} during an execution, which measures the ratio of active warps to the max possible, exceeds $98\%$. This means that the warps within a block, as well as the blocks in the grid, are performing an almost perfectly balanced workload. All these make {\textit{CodeGen-Hybrid}} significantly faster than state-of-the-art CPU/GPU baselines. 

\subsection{Overhead of Kernel Generation}
The kernel generation process, i.e., the ordering, partitioning, creation of the codes shown in the listings,  etc., is fully automated. There is no manual intervention. There is a script that gets the input matrix, automatically generates the kernels, compiles them, builds the matrix-specific executable, runs it and outputs the permanent. The whole process until execution takes less than 2 seconds for all the matrices used in the experiments. Considering the minimum execution time of \textit{CodeGen-Hybrid} being at least 478 seconds on these matrices, the overhead is negligible. 

\begin{table}[htbp]
\caption{\small Experiments on real-life matrices across two different architectures. \textit{CPU-SparsePerman} is executed on {\bf Arch-1}. Nevertheless, it could not produce a result on $curtis54$ in 3 days.}\label{tab:real_life_experiments}
\small
\centering
\scalebox{0.83}{
\tabcolsep=0.09cm
\begin{tabular}{l||rrr||rrr}
& \multicolumn{3}{c|}
{\cellcolor{gray!20}$bcsstk01$} 
& \multicolumn{3}{c}
{\cellcolor{gray!20}$mycielskian6$}
\\
& Time & \#threads & Speed. & Time & \#threads & Speed. \\
\midrule
\textit{CPU-SparsePerman}  & 34202.45 & 112 & $38.7\times$ & 15123.73 & 112 & $31.7\times$  \\
\textit{GPU-SparsePerman}  & 7570.47 & 34560 & $8.6 \times$ & 3556.81 & 34560 & $7.5\times$ \\ 
\textit{CodeGen-Hybrid} & 882.94 & 110592 & $1.0\times$ & 477.75 & 96768 & $1.0\times$ \\ 
\bottomrule
& \multicolumn{3}{c|}
{\cellcolor{gray!20}$mesh1e1$} 
& \multicolumn{3}{c}
{\cellcolor{gray!20}$d\_ss$}
\\
& Time & \#threads & Speed. & Time & \#threads & Speed. \\
\midrule
\textit{CPU-SparsePerman} & 33013.44 & 112 & $38.4\times$ & 191395.10 & 112 & $8.9\times$\\
\textit{GPU-SparsePerman}  & 6044.82 & 41472 & $7.0\times$ & 44732.08 & 34560 & $2.1\times$\\ 
\textit{CodeGen-Hybrid} & 860.19 & 124416 & $1.0\times$ & 21448.17 & 138240 & $1.0\times$ \\
\bottomrule
& \multicolumn{3}{c|}
{\cellcolor{gray!20}$curtis54$} 
& \multicolumn{3}{c}
{\cellcolor{gray!20}$bcspwr02$}
\\
& Time & \#threads & Speed. & Time & \#threads & Speed. \\
\midrule
\textit{CPU-SparsePerman}  & Timeout & 112 & $-$& 9239.81 & 112 & $7.0\times$\\
\textit{GPU-SparsePerman} & 126047.62 & 34560 & $2.6\times$& 2424.19 & 41472& $1.8\times$ \\ 
\textit{CodeGen-Hybrid} & 47841.18 & 124416 & $1.0\times$& 1317.42 & 138240& $1.0\times$ \\
\end{tabular}
}
\end{table}

\section{Related Work}

 One of the observations exploited in the literature for sparse matrices is that if the number of perfect matchings in the corresponding bipartite graph is small, then the permanent can be computed by enumerating all these matchings~\cite{mittal01}. Using CRS/CCS for sparse matrix permanents is also suggested in the same study. Another approach to exploit extreme sparsity is that one can decompose the matrices when a row/column having at most four nonzero exists~\cite{forbert03}. In the extreme case with a row/column containing a single nonzero, $n$ can be reduced by one by removing that row/column and multiplying the output by the value of the removed nonzero. A similar, tree-based decomposition is also proposed in~\cite{liang06}. Unfortunately, there are not many real-life sparse matrices that can benefit from these techniques whereas the approach proposed in this work is general and can be applied to all sparse matrices. In addition, even when decomposition is applied, the matrices will contain at least 5 nonzeros per row/column and still be sparse. Hence, \textit{CodeGen-Hybrid} can be used efficiently on them. 

In terms of parallelization, computing sparse matrix permanents is also studied in the literature~\cite{kaya19}. We used the base algorithm along with the proposed optimizations in \textit{CPU-SparsePerman} and efficiently ported the algorithm to GPUs, which we also used as a baseline, \textit{GPU-SparsePerman}. We are not aware of a study that utilizes GPUs to compute the permanents of sparse matrices. 

\section{Conclusion and Future Work}
In this work, we proposed two code-generation techniques for computing sparse matrix permanents. The latter achieved a GPU utilization exceeding $98\%$ in certain scenarios which made it significantly faster than state-of-the-art CPU/GPU baselines. This utilization level was made possible by efficiently leveraging GPU registers, reducing memory access latencies, and allowing the schedulers to almost always find an available warp to schedule onto SM cores. Moreover, we exploited a key structure in the state-of-the-art algorithms, Gray code, and efficiently reduced the number of registers used by up to $70\%$. This optimization allowed a larger number of threads to concurrently remain active within the grid. 

It is straightforward to extend code generation to multiple GPUs/nodes since permanent computation is pleasingly parallel. In the future, we aim to utilize shared memory in which to spill our registers intelligently, alongside global memory, for even faster access. This will require a multi-level ordering/partitioning algorithm. In addition, as the experiments on real matrices show, exploiting the sameness of nonzero values is extremely promising since this makes the $\xb$ array contain at least one zero in many consecutive iterations. We believe that this is an interesting research avenue to make \textit{CodeGen-Hybrid} even more efficient. 

\bibliographystyle{IEEEtran}
\bibliography{permanent.bib}

\begin{IEEEbiography}
[{\includegraphics[width=1in,height=1.25in,clip,keepaspectratio]{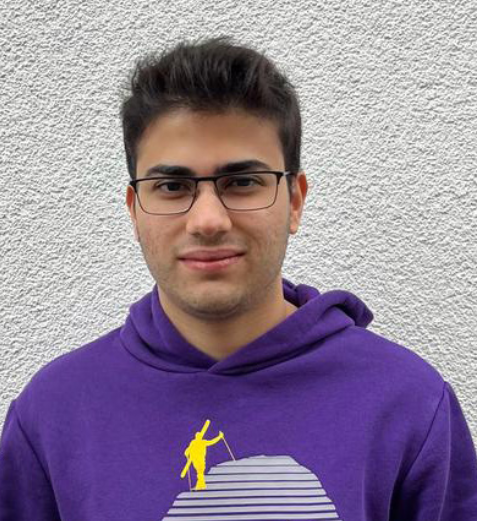}}]{Deniz Elbek} is an undergraduate student in the Computer Science and Engineering program at Sabanci University. His areas of research focus on High Performance Computing, Parallel Processing and Algorithms, Parallel Computer Architecture, and Hardware-Software Interface.
\end{IEEEbiography}

\begin{IEEEbiography}[{\includegraphics[width=1in,height=1.25in,clip,keepaspectratio]{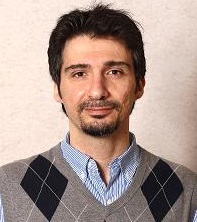}}]{Kamer Kaya} is an Associate Professor at Sabanci University. His research interests are Parallel Algorithms, High Performance Computing, and Graph and Sparse Matrix Algorithms.
\end{IEEEbiography}

\end{document}